%% file: PredTrigJournal.tex
\documentclass[journal]{IEEEtran}
\usepackage[cmex10]{amsmath}
\usepackage{mathtools}
\usepackage{amsfonts,amssymb,bm}
\usepackage{amsthm}

\usepackage{graphicx,color,psfrag}		% required for LaPrint
\usepackage{cite}			% produce ranges when there are three or more consecutive citation numbers
\usepackage[hidelinks]{hyperref}
\usepackage[capitalise]{cleveref}
\crefname{equation}{}{} % Default to eqref for equations
\crefname{section}{Sec.}{Sec.}

\usepackage{epstopdf}	% for including eps files
\ifCLASSOPTIONcompsoc
    \usepackage[caption=false,font=normalsize,labelfont=sf,textfont=sf]{subfig}
\else
    \usepackage[caption=false,font=footnotesize]{subfig}
\fi

\graphicspath{{.},{figures/},{Figures_Dominik/}}

\usepackage{todonotes}
\input{defs}

\usepackage[per=frac]{siunitx}
\usepackage{pgfplots}
\pgfplotsset{compat=newest,unit code/.code={\si{#1}},plot coordinates/math
parser=false,grid style={lightgray},tick label style={font=\footnotesize},
label style={font=\footnotesize},
legend style={font=\footnotesize}}
\usepgfplotslibrary{units}

\usepackage{tikz}
\usetikzlibrary{shapes,arrows,positioning,decorations.pathreplacing,calc,spy,angles,quotes,patterns,babel,fit}
\usepgfplotslibrary{external,groupplots} 
\tikzsetexternalprefix{Figures_Dominik/}

\makeatletter
\renewcommand{\todo}[2][]{\tikzexternaldisable\@todo[#1]{#2}\tikzexternalenable}
\makeatother

\usepackage{fancyhdr}
\newcommand{\mytitle}{\textbf{Accepted final version.}
Article accepted for publication.  To appear in \textit{IEEE Internet of Things Journal}.\\
\copyright 2019 IEEE. Personal use of this material is permitted. Permission from IEEE must be obtained for all other uses, in any current or future media, including reprinting/republishing this material for advertising or promotional purposes, creating new collective works, for resale or redistribution to servers or lists, or reuse of any copyrighted component of this work in other works.
}
\begin{document}

% Suggestions:
% ...

\title{Resource-aware IoT Control: Saving Communication through Predictive Triggering}

\author{Sebastian~Trimpe,~\IEEEmembership{Member,~IEEE}, and
  Dominik Baumann
   %Michael~Shell,~\IEEEmembership{Member,~IEEE,} 
   %John~Doe,~\IEEEmembership{Fellow,~OSA,}
   %and~Jane~Doe,~\IEEEmembership{Life~Fellow,~IEEE}% <-this % stops a space
\thanks{S.~Trimpe and D.~Baumann are with the Intelligent Control Systems Group at the Max Planck
Institute for Intelligent Systems, 70569 Stuttgart, Germany. E-mail: trimpe@is.mpg.de, dbaumann@tuebingen.mpg.de.}% <-this % stops a space
\thanks{This work was supported in part by the German Research Foundation (DFG) Priority Program 1914 (grant TR 1433/1-1), the Max Planck ETH Center for Learning Systems, the Cyber Valley Initiative, and the Max Planck Society.}%
}

% ARXIV-change: commented following 2 lines
%\markboth{Draft manuscript}%
%{Draft manuscript}


\maketitle

%%%%%%%%%%%%%%%%%%%%%%%%%%%%%%%%%%%%%%%%%%%%%
% ARXIV-change (added after final submission)
%%%%%%%%%%%%%%%%%%%%%%%%%%%%%%%%%%%%%%%%%%%%%
\thispagestyle{fancy}	% final submitted: empty
\pagestyle{empty}
%%%%%%%%%%%%%%%%%%%%%%%%%%%%%%%%%%%%%%%%%%%%%
%%%%%%%%%%%%%%%%%%%%%%%%%%%%%%%%%%%%%%%%%%%%%

\input{content/0_Abstract.tex}

% \fi
\IEEEpeerreviewmaketitle

%%%%%%%%%%%%%%%%%%%%%%%%%%%%%%%%%%%%%%%%%%%%%%%%%%%%%%%%%%%%%%%%%%%%%%%%%%%%%%%%
% Sections:
%\input{parts/introduction}

\input{content/1_Introduction.tex}
\input{content/RelatedWork.tex}

%\input{content/Notation.tex}

\input{content/3_TriggeringProblem.tex}

\input{content/4_TriggeringFramework.tex}

\input{content/5_PTST.tex}

\input{content/6_TriggerAnalysis.tex}

\input{content/7_Experiments.tex}

\input{content/MultiAgentScenario.tex}

\input{content/8_Simulations.tex}

%\input{content/ConnectionLWB.tex}  % removed, discussed in conclusion + letter

\input{content/9_Conclusion.tex}

%%%%%%%%%%%%%%%%%%%%%%%%%%%%%%%%%%%%%%%%%%%%%%%%%%%%%%%%%%%%%%%%%%%%%%%%%%%%%%%
\appendices

\section{Proof of Lemma \ref{lem:PDF_eII}}
\label{app:proofeII}
\input{content/Proof2.tex}

\section{Proof of Lemma \ref{lem:PDF_eI}}
\label{app:proofeI}
\input{content/Proof1.tex}

%%%%%%%%%%%%%%%%%%%%%%%%%%%%%%%%%%%%%%%%%%%%%%%%%%%%%%%%%%%%%%%%%%%%%%%%%%%%%%%
\section*{Acknowledgment}
The authors thank
%are grateful to 
their colleagues 
Felix Grimminger and Alonso Marco for their support with the experimental setup
and Friedrich Solowjow for insightful discussions.

%%%%%%%%%%%%%%%%%%%%%%%%%%%%%%%%%%%%%%%%%%%%%%%%%%%%%%%%%%%%%%%%%%%%%%%%%%%%%%%%
\bibliographystyle{IEEEtran}
\bibliography{Database}

%%%%%%%%%%%%%%%%%%%%%%%%%%%%%%%%%%%%%%%%%%%%%%%%%%%%%%%%%%%%%%%%%%%%%%%%%%%%%%%%
\begin{IEEEbiography}[{\includegraphics[width=1in,height=1.25in,clip,keepaspectratio]{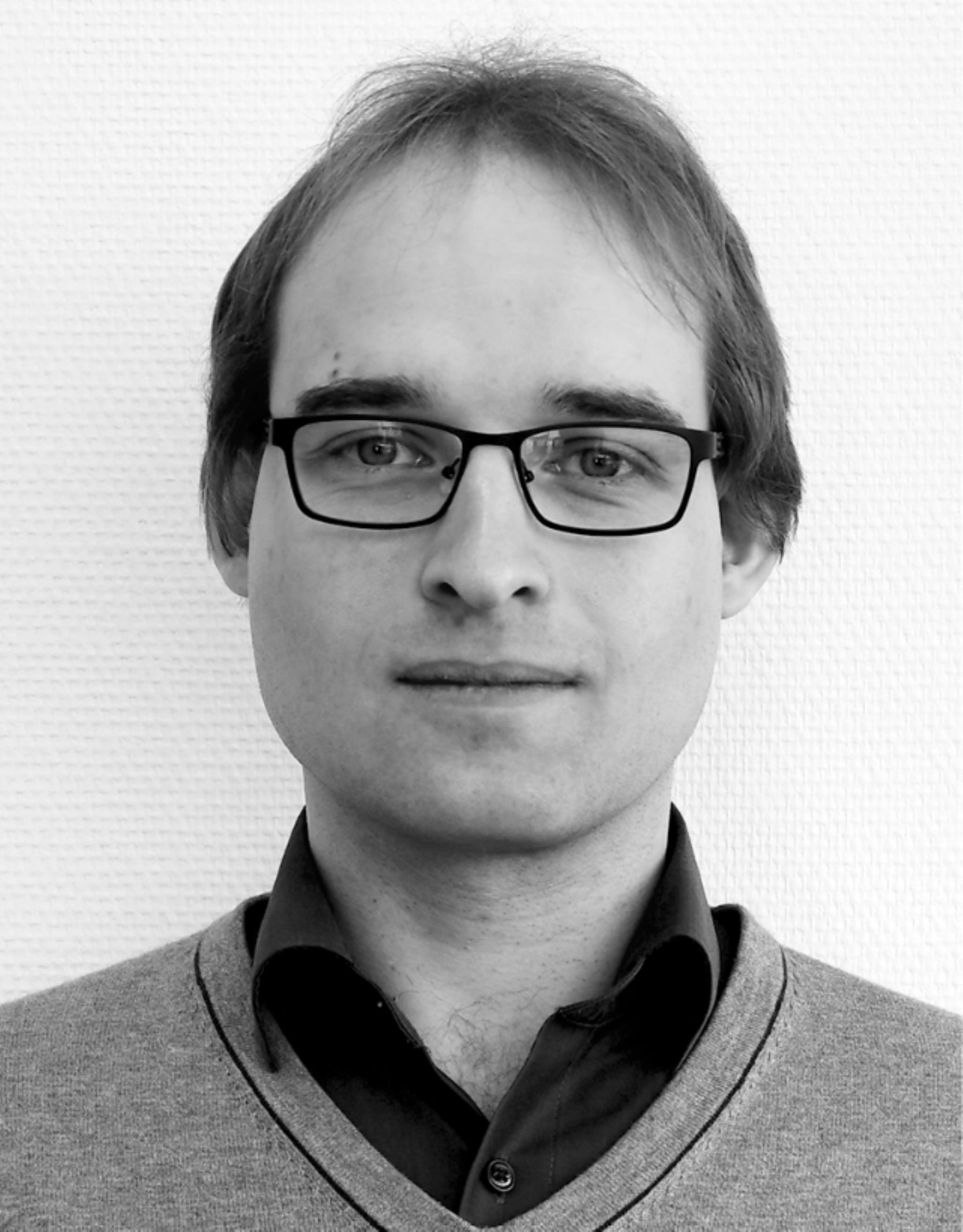}}]{Sebastian Trimpe}
% previous: SebastianTrimpeBW
(M'12) 
received the B.Sc. degree
in general engineering science and the M.Sc.
degree (Dipl.-Ing.) in electrical engineering from
Hamburg University of Technology, Hamburg,
Germany, in 2005 and 2007, respectively, and the
Ph.D. degree (Dr. sc.) in mechanical engineering
from ETH Zurich, Zurich, Switzerland, in 2013.

He is currently a Research Group Leader at the Max Planck Institute for Intelligent Systems, Stuttgart, Germany, where he leads the independent Max Planck Research Group on Intelligent Control Systems.
%Previously, he was Lecturer and Postdoctoral Researcher
%at ETH Zurich. In 2007, he was a research scholar at the University of
%California at Berkeley. 
His main research interests are in systems and control theory, machine learning, networked and autonomous systems.
%His research interests are in the area of control systems
%theory and design with emphasis on autonomous, networked, and learning
%systems.

Dr. Trimpe is a recipient of the General Engineering Award for the best undergraduate
degree (2005), a scholarship from the German National Academic
Foundation (2002 to 2007), the triennial IFAC World Congress Interactive
Paper Prize (2011), and the Klaus Tschira Award for public understanding of science (2014).
\end{IEEEbiography}

% Sebastian (initial submission):
%\begin{IEEEbiography}[{\includegraphics[width=1in,height=1.25in,clip,keepaspectratio]{SebastianTrimpeBW}}]{Sebastian Trimpe}
%(M’12) received the B.Sc. degree
%in general engineering science and the M.Sc.
%degree (Dipl.-Ing.) in electrical engineering from
%Hamburg University of Technology, Hamburg,
%Germany, in 2005 and 2007, respectively, and the
%Ph.D. degree (Dr. sc.) in mechanical engineering
%from ETH Zurich, Zurich, Switzerland, in 2013.
%
%He is currently a Research Group Leader in the Autonomous
%Motion Department at the Max Planck Institute for Intelligent Systems, T\"bingen, Germany, where he leads the Intelligent Control Systems group.
%His research interests are in the area of control systems
%theory and design with emphasis on autonomous, networked, and learning
%systems.
%
%Dr. Trimpe is a recipient of the General Engineering Award for the best undergraduate
%degree (2005), a scholarship from the German National Academic
%Foundation (2002 to 2007), the triennial IFAC World Congress Interactive
%Paper Prize (2011), and the Klaus Tschira Award for public understanding of science (2014).
%\end{IEEEbiography}

\begin{IEEEbiography}[{\includegraphics[width=1in,height=1.25in,clip,keepaspectratio]{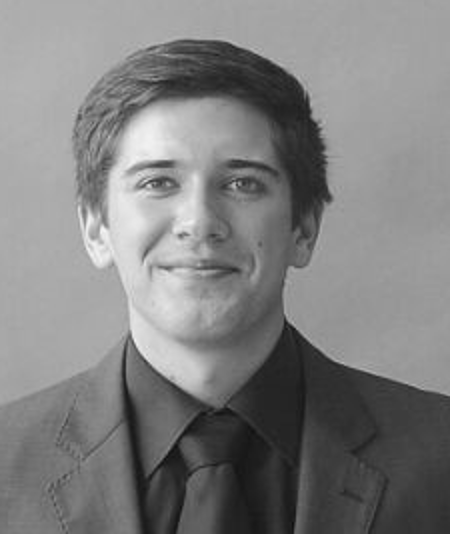}}]{Dominik
Baumann} received the Dipl.-Ing. degree in electrical engineering from TU Dresden,
Germany, in 2016. He is currently a PhD student in the Intelligent Control Systems Group at the Max Planck Institute for Intelligent Systems, T\"ubingen,
Germany.

His research interests include control theory, robotics, distributed and
cooperative control, learning and networked control systems.
\end{IEEEbiography}


\vfill

\end{document}

%% file: defs.tex
\newcommand{\vik}[3]{{#1}^{#2}_{#3}}        % variable, i, k
\newcommand{\xik}[2]{\vik{x}{#1}{#2}}
\newcommand{\uik}[2]{\vik{u}{#1}{#2}}

\newcommand{\Uall}{\mathcal{U}}

\newcommand{\xKFik}[2]{\hat{x}^{#1}_{#2}}
\newcommand{\PKFik}[2]{P^{#1}_{#2}}
\newcommand{\xKF}{\hat{x}}
\newcommand{\PKF}{P}
\newcommand{\PKFss}{\bar{P}}

\newcommand{\xPik}[2]{\check{x}^{#1}_{#2}}        % prediction / remote estimate
\newcommand{\xP}{\check{x}}

\newcommand{\ie}{i\/.\/e\/.,\/~}
\newcommand{\eg}{e\/.\/g\/.,\/~}
\newcommand{\cf}{cf\/.\/~}
\newcommand{\fig}{Fig\/.\/~}

\newcommand{\sect}{Sec\/.\/~}

\newcommand{\Lem}{Lemma~}

\newcommand{\Ex}{Example~}

\newcommand{\transp}{\text{T}}

\newtheorem{lemma}{Lemma}

\newtheorem{proposition}{Proposition}
\newtheorem{corollary}{Corollary}
\newtheorem{assumption}{Assumption}

\newtheorem{example}{Example}
\newtheorem{remark}{Remark}

\newcommand{\Yall}{\mathcal{Y}}
\newcommand{\Gamall}{\Gamma}
\newcommand{\eKF}{\hat{e}}
\newcommand{\eIaux}{e^{\mathrm{nc}}}

\newcommand{\eIhat}{\hat{e}^{\mathrm{nc}}}
\newcommand{\eIIhat}{\hat{e}^{\mathrm{c}}}
\newcommand{\PI}{P^{\mathrm{nc}}}
\newcommand{\PII}{P^{\mathrm{c}}}

\newcommand{\notrigk}{|_{\gamma_k=0}}
\newcommand{\trigk}{|_{\gamma_k=1}}
\newcommand{\notrigvar}[1]{|_{\gamma_{#1}=0}}
\newcommand{\trigvar}[1]{|_{\gamma_{#1}=1}}

\newcommand{\Vo}{V_\text{o}}

\newcommand{\Voi}[1]{V_{\text{o}}^{#1}}

\newcommand{\commC}{C} 	% communication cost (one-step)
\newcommand{\last}{\ell}			% last instant (in past), where triggered	
\newcommand{\lastel}{\kappa}		% time-index of last element =1 in \Gamall_k

\newcommand{\trigsig}{\bar{E}}
\newcommand{\trigsigM}{\bar{E}^\text{mean}}
\newcommand{\trigsigV}{\bar{E}^\text{var}}

\newcommand{\Cc}{\mathcal{C}}
\newcommand{\Dc}{\mathcal{D}}
\newcommand{\Ec}{\mathcal{E}}

\newcommand{\Nc}{\mathcal{N}}

\DeclareMathOperator*{\E}{\field{E}}
\DeclareMathOperator*{\Var}{Var}
\DeclareMathOperator*{\trace}{trace}
\DeclareMathOperator*{\diag}{diag}

\newcommand{\field}[1]{\mathbb{#1}}
\newcommand{\R}{\field{R}}
\newcommand{\N}{\field{N}}

\providecommand{\norm}[1]{\lVert#1\rVert}

\newcommand{\graph}[1]{{\bf #1}}

%% file: content/0_Abstract.tex
\begin{abstract}
The Internet of Things (IoT) interconnects multiple physical devices in large-scale networks.
When the `things' coordinate decisions and act collectively on shared information, \emph{feedback} is introduced between them.  
Multiple feedback loops are thus closed over a shared, general-purpose network.  
Traditional feedback control is unsuitable for design of IoT control because it relies on high-rate periodic communication and is ignorant of the shared network resource.  
Therefore, recent event-based estimation methods are applied herein for resource-aware IoT control allowing agents to decide online whether communication with other agents is needed, or not.  While this can reduce network traffic significantly, a severe limitation of typical event-based approaches is the need for instantaneous triggering decisions that leave no time to reallocate freed resources (\eg communication slots), which hence remain unused.  
To address this problem, novel predictive and self triggering protocols are proposed herein.  From a unified Bayesian decision framework, two schemes are developed: self triggers that predict, at the current triggering instant, the next one; and predictive triggers that check at every time step, whether communication will be needed at a given prediction horizon.  The suitability of these triggers for feedback control is demonstrated in hardware experiments on a cart-pole, and scalability is discussed with a multi-vehicle simulation.
\end{abstract}

\begin{IEEEkeywords}
Internet of Things,
%IoT control,
%feedback,
feedback control,
event-based state estimation,
predictive triggering,
self triggering,
distributed control,
resource-aware control.
\end{IEEEkeywords}

%% file: content/1_Introduction.tex
\section{Introduction}
\label{sec:intro}
The Internet of Things (IoT) will connect large numbers of physical devices via local and global networks, \cite{AtIeMo10,GuBuMaPa13}.  While early IoT research concentrated on problems of data collection, communication, and analysis \cite{Sa16}, 
%with important application such as wireless sensor networks and mobile apps \Sebtodo{refs}, 
%`closing the loop' by using the available data for \emph{actuation}
using the available data for actuation 
%and feedback 
is vital for envisioned applications such as autonomous vehicles, building automation, or cooperative robotics.
% data collection (\eg sensing), distribution (communication, mobile applications), and processing (data analysis) \cite{}, 
In these applications, the devices or `things' are required to \emph{act} intelligently based on data from local sensors and the network. 
For example, cars in a platoon need to react to other cars' maneuvers 
%such as braking 
to keep a desired distance;
% down or up the platoon to keep constant distance; 
and climate control units must coordinate their action for optimal ambience in a large building.
 \emph{IoT control} thus refers to systems where data about the physical processes, collected via sensors and communicated over networks, are used to decide on actions.
 %, \eg how much a car should accelerate, or an climate device should cool down.  
 These actions in turn affect the physical processes, which is the core principle of \emph{closed-loop control} or \emph{feedback}.
% omitted for space: \cite{AsMu10}.

% Figure removed for space considerations:
%
%\begin{figure}[tb]
%\centering
%\subfloat[Autonomous vehicle platoon]{%
%    \input{Tikz/PlatooningSchematic.tex}%
%    \label{fig:sub:iotExVehicles}}
%\\
%\subfloat[Smart building??]{%
%    \input{Tikz/smart_home.tex}%
%    \label{fig:sub:iotExTwo}}
%\caption{IoT control example scenarios.  Feedback loops are closed over the network: the cars in (a) need to react to other cars' maneuvers such as braking down or up the platoon to keep constant distance; the cooling/heating devices in (b) need to coordinate for optimal building climate.
%\Seb{Dominik, could you have a crack at (b) when you have time?}}
%\label{fig:iotExamples}
%\end{figure}

Figure \ref{fig:iotControlSchematic} shows an abstraction of a general IoT control system.  
%Most IoT work to date concerns problems without actuation (no red arrows in \fig \ref{fig:iotControlSchematic}); that is, the data and computations by the agents and the network do not influence the physical processes.  
When the available information within the IoT is used for decision making and commanding actuators (red arrows), one introduces feedback between the cyber and the physical world, \cite{Sa16}. 
% Omitted for space:  AsMu10
%This problem is arguably harder because feedback changes the dynamics of the overall systems, which can enable high-performance applications, but also lead to failures when not done right (\eg destabilizing a naturally stable system), \cite{Sa16,AsMu10}.
%be for the better, or worse \cite{}.  
Feedback loops can be closed on the level of a local object, but, more interestingly, also across agents and networks.  
Coordination among agents is vital, for example, when agents seek to achieve a global objective.
IoT control aims at enabling coordinated action among multiple things.
%which is essential in the applications mentioned above, as well as many others.

\begin{figure}[tb]
\centering
\includegraphics[width=\columnwidth]{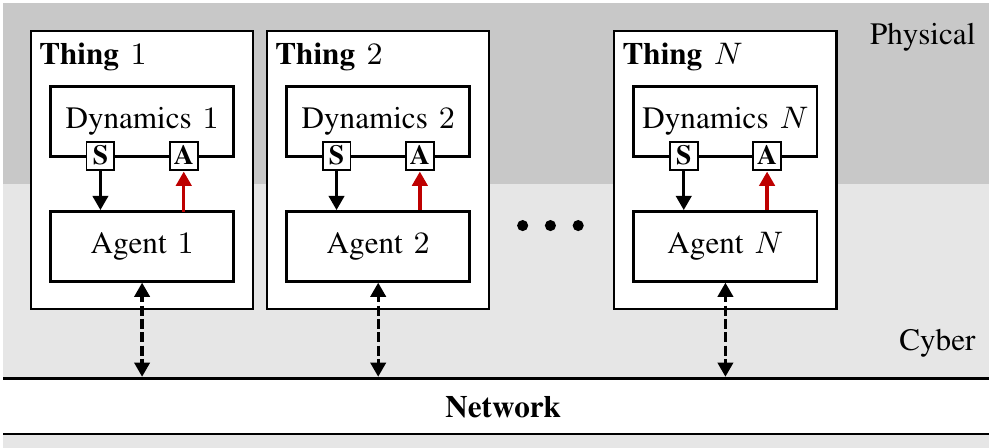}
\caption{Abstraction of an IoT control system. Each \emph{Thing} is composed of \emph{Dynamics} representing its physical entity and an \emph{Agent} representing its algorithm unit.  Dynamics and Agent are interconnected via sensors (S) and actuators (A).  The \emph{Network} connects all things to the IoT.}
\label{fig:iotControlSchematic}
\end{figure}

In contrast to traditional feedback control systems, where feedback loops are
closed over dedicated communication
lines (typically wires), feedback loops in IoT control are realized over a
general purpose network such as the Internet or local networks.  In typical IoT applications, these networks are wireless. While networked communication offers great advantages in terms
of, inter alia, reduced installation costs, unprecedented flexibility, and availability of data, control over networks involves formidable challenges for system design and operation, for example, because of imperfect communication, variable network structure, and limited communication resources, \cite{Lu14,HeNaXu07}.  Because the network bandwidth is shared by multiple entities, each agent should use the communication resource \emph{only when necessary}.  Developing such resource-aware control for the IoT is the focus of this work.
%network access designed in tandem with control strategies.  
This is in contrast to traditional feedback control, where data transmission typically happens periodically at a priori fixed update rates.
%, irrespective of whether there is need for an update or not.
%the information content of the data or state of the sys.
%and which is hence not suitable for resource-limited IoT control.

Owing to the shortcomings of traditional control, \emph{event-based methods} for state estimation and control have emerged since the pioneering work \cite{AsBe99,Ar99}.  The key idea of event-based approaches is to apply feedback only upon certain \emph{events} indicating that transmission of new data is necessary (\eg a control error passing a threshold level, or estimation uncertainty growing too large).  Core research questions concerning the design of the event triggering laws, which decide when to transmit data, and the associated estimation and control algorithms with stability and performance guarantees have been solved in recent years 
% cut (space): with closed-loop guarantees on stability, performance and robustness despite limiting the communication 
(see \cite{Le11,HeJoTa12,Mi15,ShShCh16} for overviews).
%Several approaches for event-based control and state estimation have been developed since (see \cite{Le11,HeJoTa12,Mi15,ShShCh16} for overviews) 
%
%In recent years, the research community in event-based control and state estimation has had 

This work builds on a framework for distributed event-based state estimation (DEBSE) developed in prior work \cite{TrDAn11,TrDAn14b,Tr17,MuTr18}, which is applied herein to resource-aware IoT control as in \fig \ref{fig:iotControlSchematic}.  
The key idea of DEBSE is to employ model-based predictions of other things to avoid the need for continuous data transmission between the agents.  Only when the model-based predictions become too inaccurate (\eg due to a disturbance or accumulated error), an update is sent. 
%DEBSE thus provides an architecture for resource-efficient networked control, which employs, 
%on each agent, (i) state estimators to predict information of relevant
%other agents at times when there is no communication, and (ii) event triggering
%laws that issue communication whenever predictions are not good enough.
%With this architecture, each agent has all relevant information available for making coordinated control decisions locally.
Figure \ref{fig:agentSchematic} represents one agent of the IoT control system in \fig \ref{fig:iotControlSchematic} and depicts the key components of the DEBSE architecture:
%The main components implemented on each agent are:
%depicts the components of the DEBSE architecture; it represents one of the agents of the IoT control system in \fig \ref{fig:iotControlSchematic}.  The 
%key ideas to
%maintain/provide 
%make relevant information from other agents available on each agent, yet limit inter-agent communication, 
%can be summarized as follows: 
% of this architecture, which that maintaints/provides relevant information on each agent, yet limits communication between agents
\begin{itemize}
\item \emph{Local control:} Each agent makes local control decisions for its actuator; for coordinated action across the IoT, it also needs information from other agents in addition to its local sensors.
\item \emph{Prediction of other agents:} State estimators and predictors 
%based on agents' dynamics models 
(\eg of Kalman filter type) are used to predict the states of all, or a subset of agents based on agents' dynamics models; these predictions are reset (or updated) when new data is received from the other agents.
%Sensor measurements (or local states) of all agents (or a relevant subset of these) are available locally via state prediction based on agents' dynamics models; estimates and predictions
\item \emph{Event trigger:} Decides when an update is sent to all agents in the IoT.  For this purpose, the local agent implements a copy of the predictor
of its own behavior (\emph{Prediction Thing $i$}) to replicate locally the information the other agents have about itself.  The event trigger compares the prediction with the local state estimate: the current state estimate is transmitted to other agents only if the prediction is not sufficiently accurate.
\end{itemize}
Key benefits of this architecture are: each agent has all relevant information available for coordinated decision making, but inter-agent communication is limited to the necessary instants (whenever model-based predictions are not good enough).
%This architecture enables effective distributed decision making with limited communication by providing each agent with all relevant information, while 

\begin{figure}[tb]
\centering  
\includegraphics{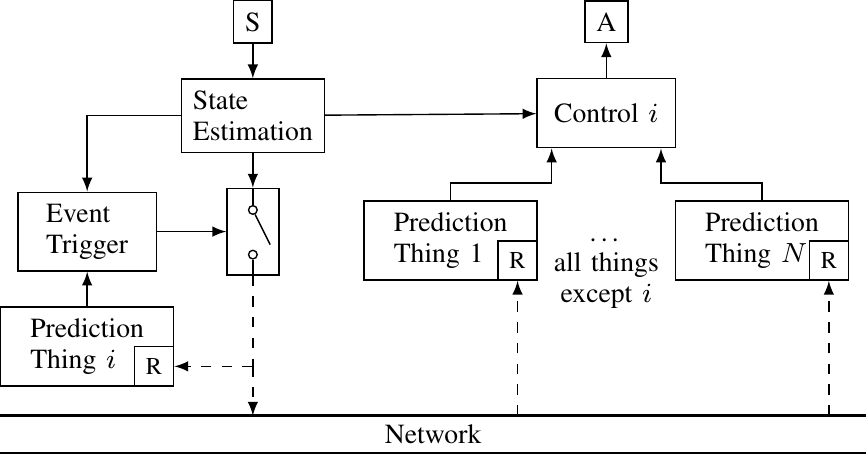}
\setlength{\abovecaptionskip}{-0.25cm}
\caption{Algorithmic components implemented on each agent
$i=1,\dots,N$ of the IoT control system in \fig \ref{fig:iotControlSchematic}. Agent $i$'s control decision is based on local information (\emph{State Estimation}) and predictions of all (or a subset of) other things (\emph{Prediction Thing 1 to $N$}).  Each agent sends an update (\emph{Event Trigger}) to all other agents whenever the prediction of its own state (\emph{Prediction Thing $i$}) deviates too far from the truth, so that predictions can be reset (R).}
\label{fig:agentSchematic}
\end{figure}

Experimental studies \cite{TrDAn11,Tr17} demonstrated that DEBSE can achieve significant communication savings,
% compared to periodic sampling, 
which is inline with many other studies in event-based estimation and control.
% (such as stabilizing an unstable system \cite{TrDAn12b}). 
The research community has had remarkable success in showing
% -- both theoretically and experimentally -- 
that the number of samples in feedback loops can be reduced significantly as compared to traditional time-triggered designs.  
%The resulting reduction in average communication or processing 
This can be translated into increased battery life \cite{ArMaAnTaJo13} in wireless sensor systems, for example.
Despite these successes, better utilization of shared communication resources has typically \emph{not} been demonstrated.
%it has rarely been demonstrated that event-based designs also result in better utilization of shared communication and processing resources, or reduced hardware costs.  
%The latter is, however, often the ultimate goal and reason for considering event-based designs.  
%
A fundamental problem of most event-triggered designs (incl.\ DEBSE) is that they make decisions about whether a communication 
%or control computation 
is needed \emph{instantaneously}.
%instantaneously whether a communication or control computation is needed \emph{now}.  
This means that the resource must be held available at all times in case of a positive triggering decision.  Conversely, if a triggering decision is negative, the reserved slot 
%(\eg in distributed communication network) 
remains unused because it cannot be reallocated to other users immediately.
%While secondary savings from reduced communication or processing have been realized (\eg increased battery life\cite{ArMaAnTaJo13}), the primary goal of event-triggered approaches, namely the better usage of available resources, mostly fails to be realized.  Since resources cannot effectively  be reallocated to other processes with many event-triggered approaches and thus remain unused, one may argue that traditional periodic communication is still beneficial in most applications as it at least does not forfeit available resources.

In order to translate the reduction in average sampling rates 
%, which is possible \emph{in principle}, 
to better \emph{actual} resource utilization, it is vital that the event-based system is able to \emph{predict} resource usage ahead of time, rather than requesting resources instantaneously.  %Only 
This allows the processing or communication system  to reconfigure and make unneeded resources available to other users or set to sleep for saving energy.  Developing such predictive triggering laws for DEBSE and their use for resource-aware IoT control
% (in particular, the setup in \fig \ref{fig:remoteEstimation}) 
are the main objectives of this article.
%
%In this paper, we proposed a systematic framework for deriving such predictive triggering mechanisms for the problem of state estimation and, in particular, the setup depicted in \fig \ref{fig:remoteEstimation}.
%Developing a principled way for deriving such predictive triggering mechanisms for the problem of state estimation is the object of this paper.
% the topic addressed in this paper
%For this, we consider the archetypal remote estimation problem shown in \fig \ref{fig:remoteEstimation}, where a sensor (with sufficient processing capabilities) decides whether and when to communicate its sensor data to a state estimator at a remote location.
%

%%%%%%%%%%%%%%%%%%%%%%%%%%%%%%%%%%%%%%%%%%%%%%%%%%%%%%%%%%%%%%%%%%%%%%%%%%%%%%%%
\subsubsection*{Contributions}
This article proposes a framework for resource-aware IoT control based on DEBSE.
%This article presents a framework and algorithms for resource-aware IoT control, where physical objects share data with each other to coordinate their actions while limiting the use of communication resources to the necessary instants.  The main focus are predictive triggering mechanisms, by which agents can predict future data transmissions, which shall enable efficient (re-)scheduling of network resources.
The main contributions are summarized as follows:
\begin{enumerate}
%[leftmargin=5mm,label={--}]
%\item Review of distributed event-based state estimation (DEBSE) for resource-aware IoT control;
%efficiently distributing relevant information in the IoT enabling, among others, cooperative feedback control;

%\item Proposal of predictive triggering 
%for state estimation
%to predict future communication instants;

\item Proposal of a Bayesian decision framework for deriving predictive triggering mechanisms, which provides a new perspective on the triggering problem in estimation;

%Extending previous work \cite{TrCa15} on event trigger design, we propose a unified decision framework for developing different predictive triggering mechanisms, where triggering is formulated as an optimization problem solved under different information patterns.  To the best of the author's knowledge, this provides a new perspective on the triggering problem in estimation.  The framework is used to develop the following two triggering concepts.  

\item Derivation of two novel triggers from this framework:  the \emph{self trigger}, which predicts the next triggering instant based on information available at a current triggering instant; and the \emph{predictive trigger}, which predicts triggering for a given future horizon of $M$ steps;
%a trigger $M\!>\!0$ steps ahead of time, where the prediction horizon $M$ is a design parameter; 

%\item First, a \emph{self triggering} rule is derived that predicts the next trigger based on the information available at a current triggering instant. The self trigger is closely related to the concept of variance-based triggering \cite{TrDAn14b}, albeit this concept has not been used for self triggering before.

%\item Second, we propose and develop the concept of \emph{predictive triggering}.  In contrast to self triggering, where the next trigger is computed at the last triggering instant, the predictive trigger continuously monitors the sensor measurements, but predicts a communication $M\!>\!0$ steps ahead of time, where the prediction horizon $M$ is a design parameter.
%Predictive triggering is a novel concept, which is situated between the concepts of event triggering and self triggering.

\item Demonstration and comparison 
%Demonstration 
of the proposed triggers in experiments on an inverted pendulum testbed; 
%in terms of performance/communication trade-offs, and effectiveness for feedback control;
%with comparison 
% comparing their effectiveness in trading off 
%of their estimation/communication performance and effectiveness for feedback control; 
and
%The effectiveness of the different triggers in trading off estimation performance for communication is compared in numerical simulations.

\item Simulation study of a multi-vehicle system.
\end{enumerate}

The Bayesian decision framework extends previous work \cite{TrCa15} on event trigger design to the novel concept of predicting trigger instants. The proposed self trigger is related to the concept of variance-based triggering \cite{TrDAn14b}, albeit this concept has not been used for self triggering before.  To the best of the authors' knowledge, predictive triggering is a completely new concept in both event-based estimation and control.  Predictive triggering is shown to reside between the known concepts of event triggering and self triggering.  

A preliminary version of some results herein was previously published in the conference paper \cite{Tr16}.  
%In contrast to that paper, 
This article targets IoT control and has been restructured and extended accordingly.
%; it additionally 
New results beyond \cite{Tr16} include the treatment of control inputs in the theoretical analysis (\sect \ref{sec:predAndSelfTrig}), the discussion of multiple agents (\sect \ref{sec:multiAgent}),
%:; it is rewritten and includes, in addition, 
%a discussion of DEBSE for IoT control (\sect \ref{sec:DEBSE}), 
%new theoretical results also treating control inputs (\sect \ref{sec:predAndSelfTrig}),  
hardware experiments (\sect \ref{sec:experiments}), and a new multi-vehicle application example (\sect \ref{sec:simulationStudy}).  

%% file: content/RelatedWork.tex
\section{Related Work}
\label{sec:relatedWork}
Because of the promise to achieve high-performance control on resource-limited systems,
% such as IoT with a shared communication network,
the area of event-based control and estimation has seen substantial growth in the last decades.  For general overviews, see \cite{Le11,HeJoTa12,Lu14,Mi15} for control and \cite{Le11,TrCa15,SiNoLaHa16,ShShCh16} for state estimation.  
% removed: Tr15arxiv, GrHiJuEtAl14
This work mainly falls in the category of event-based state estimation (albeit state predictions and estimates are also used for feedback, \cf \fig \ref{fig:agentSchematic}).

Various design methods have been proposed in literature for event-based state estimation and, in particular, its core components, the prediction/estimation algorithms and event triggers.  For the former, different types of Kalman filters \cite{TrDAn11,TrDAn14b,MaEsGaSa15}, 
% Shortened (due to space):
% nonlinear Kalman filters \cite{MaEsGaSa15}, 
modified Luenberger-type observers \cite{Tr17,MuTr18}, and set-membership filters \cite{SiNoHa13,ShChSh14c} have been used, for example.  Variants of event triggers include  
%have been developed that make the 
triggering based on the innovation \cite{TrDAn11,WuJiJoSh13}, estimation
variance \cite{TrDAn14b,LeDeQu15}, or entire probability density functions (PDFs)
\cite{MaSi10}.  Most of these event triggers make transmit decisions instantaneously, while the focus of this work is on predicting triggers.
% problem of \emph{predicting} future communications in EBSE has received considerably less attention.

The concept of \emph{self triggering} has been proposed \cite{VeMaFu03} to address the problem of predicting future sampling instants.  In contrast to event triggering, which requires the continuous monitoring of a triggering signal,
% cut (space): (such as a control error), 
self-triggered approaches predict the next triggering instant already at the previous trigger.
%, \cite{HeJoTa12,VeMaFu03}.  
While several approaches to self-triggered control have been proposed in literature (\eg \cite{HeJoTa12,WaLe09,MaAnTa10,AnTa10}),    
% Not included (space):
% - AlSiPa15
% - Nowzari
% - MaTa08
%
self triggering for state estimation has received considerably less attention.  Some exceptions are discussed next.
%Some of the results for estimation are briefly discussed next.

Self triggering is considered for set-valued state estimation in \cite{MePr14}, and for high-gain continuous-discrete observers in \cite{AnNaSeVi15}. In \cite{MePr14}, a new measurement is triggered when the uncertainty set about some part of the state vector becomes too large.  In \cite{AnNaSeVi15}, the triggering rule is designed so as to ensure convergence of the observer.  
The recent works \cite{BrGoHeAl15} and \cite{KoFi15} propose self triggering approaches, where transmission schedules for multiple sensors are optimized at a-priori fixed periodic time instants.
% cut (space): taking into account the cost of sampling and estimation/control performance.
While the re-computation of the schedule happens periodically, the transmission of sensor data does generally not.  
% CUT (for space reasons, and structure of this part)
%In the approach developed herein, we do not preimpose the periodic update of schedules, but determine the next triggering instant at the time of the last one.  However, as it shall be seen, periodic transmissions may results for LTI systems.
In \cite{AlSiPa12}, a discrete-time observer is used as a component of a self-triggered output feedback control system.  Therein, triggering instants are determined by the controller to ensure closed-loop stability.
%Related problems on self-triggered output feedback control, where the objective is on a property of the control loop (\eg stability) rather than estimation directly/in the first place, can be found in \cite{AlSiPa12} ... .
% not included
% - AlSiPa14: similar to AlSiPa12, but does not seem to have choice of ST scheduler
%

Alternatives to the Bayesian decision framework herein for developing triggering schedules
include dynamic programming approaches such as in~\cite{WuAr08,XuHe04,lipsa2011remote}.
% Suggestion Dominik:
%, or other optimization-based approaches as in~\cite{lipsa2011remote}. \Dom{This is the reference mentioned by the reviewer. I don't really see why exactly this reference should be mentioned here. They derive a triggering policy for a stochastic first-order system, so it is nothing really special\ldots}

None of the mentioned references considers the approach taken herein, where triggering is formulated as a Bayesian decision problem under different information patterns.  The concept of predictive triggering, which is derived from this, is novel.  It is different from self triggering in that decisions are made continuously, but for a fixed prediction horizon.  
%Predictive triggering is thus a novel concept.

%% file: content/3_TriggeringProblem.tex
\section{Fundamental Triggering Problem}
\label{sec:triggeringDecisionProblem}
\label{sec:problem}
In this section, we formulate the predictive triggering problem that each agent in \fig \ref{fig:agentSchematic} has to solve, namely predicting when local state estimates shall be transmitted to other agents of the IoT.
%For this, we focus on the fundamental triggering problem that each agent in \fig \ref{fig:agentSchematic} has to solve, namely when to transmit the local state estimates to other agents that need its state information for decision making.  
%To this end, 
We consider the setup in \fig \ref{fig:remoteEstimation}, which has been reduced to the core components required for the analysis in subsequent sections.
% of the triggering decision.  
 Agent $i$, called \emph{sensor agent}, sporadically transmits data over the network to agent $j$.
%, which needs agent $i$'s state information for decision making. 
Agent $j$ here stands representative for any of the agents in the IoT 
%in \fig \ref{fig:iotControlSchematic} 
that require information from agent $i$.
Because agent $j$ can be at a different location, it is called \emph{remote agent}. 
We next introduce the components of \fig \ref{fig:remoteEstimation}
% (\emph{dynamics}, \emph{state estimation}, \emph{prediction}), 
and then make the predictive triggering problem precise.

\begin{figure}[tb]
\centering
\includegraphics{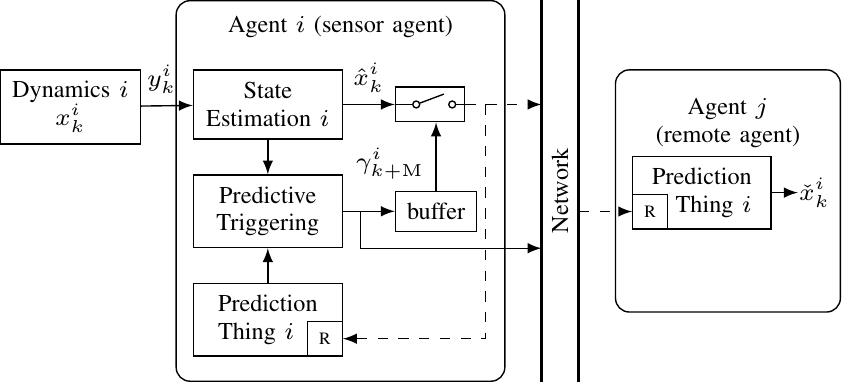}
\caption{Predictive triggering problem.  The sensor agent $i$ runs a local \emph{State Estimator} and transmits its estimate $\xKF_k^i$ to the remote agent $j$ in case of a positive triggering decision ($\gamma_k^i = 1$).  The predictive trigger computes the triggering decisions ($\gamma_{k+M}^i \in \{0,1\}$) $M$ steps ahead of time.  This information can be used by the network to allocate resources.
Local control (\cf \fig \ref{fig:agentSchematic}) is omitted here for clarity, but treated in the analysis.
%respond to the triggering requests by reallocating resources.  
%Developing predictive triggering mechanisms is the objective of this paper.
}
\label{fig:remoteEstimation}
\end{figure}

%%%%%%%%%%%%%%%%%%%%%%%%%%%%%%%%%%%%%%%%%%%%%%%%%%%%%%%%%%%%%%%%%%%%%%%%%%%%%%%%
\subsection{Process dynamics} 	
\label{sec:processDynamics}
We consider each agent $i$ to be governed by stochastic, linear dynamics with Gaussian noise,
% (a standard setting in state estimation \cite{AnMo05}),
\begin{align}
\xik{i}{k} &= A_{i} \xik{i}{k-1} + B_{i} \uik{i}{k-1} + \vik{v}{i}{k-1} 
\label{eq:sys_x}
\\
\vik{y}{i}{k} &= H_{i} \xik{i}{k} + \vik{w}{i}{k} 			 \label{eq:sys_y}
\end{align}
with $k \! \geq \! 1$ the discrete time index, 
%$i$ denoting thing $i$, 
$\xik{i}{k}
\in \R^{n_\text{x}}$ the state, $\uik{i}{k} \in \R^{n_\text{u}}$ the input,
$\vik{v}{i}{k} \in \R^{n_\text{x}}$ process noise (\eg capturing model uncertainty), $\vik{y}{i}{k} \in \R^{n_\text{y}}$ the sensor measurements, and
$\vik{w}{i}{k} \in \R^{n_\text{y}}$ sensor noise.  The random variables $\xik{i}{0}$,
$\vik{v}{i}{k}$, and $\vik{w}{i}{k}$ are mutually independent with PDFs $\Nc(\xik{i}{0};
\bar{x}_{i},X_{i})$, $\Nc(\vik{v}{i}{k}; 0,Q_{i})$, and $\Nc(\vik{w}{i}{k};
0,R_{i})$, where $\Nc(z; \mu, \Sigma)$ denotes the PDF of a Gaussian random
variable $z$ with mean $\mu$ and variance $\Sigma$. 

Equations \eqref{eq:sys_x} and \eqref{eq:sys_y} represent decoupled agents' dynamics, which we consider in this work (\cf \fig \ref{fig:iotControlSchematic}).  Agents will be coupled through their inputs (see \sect \ref{sec:control} below).
While the results are developed herein for the time-invariant dynamics \eqref{eq:sys_x}, \eqref{eq:sys_y} to keep notation uncluttered, they readily extend to the linear time-variant case (\ie $A_{i}$, $B_i$, $H_i$, $Q_i$, and $R_i$ being functions of time $k$).  
%For ease of notation, the derivations are stated for the time invariant case, but 
Such a problem is discussed in \sect \ref{sec:simulationStudy}.

The sets of all measurements and inputs up to time $k$ are denoted by $\Yall^i_k := \{ \vik{y}{i}{1}, \vik{y}{i}{2}, \dots, \vik{y}{i}{k} \}$ and $\Uall^i_k := \{ \vik{u}{i}{1}, \vik{u}{i}{2}, \dots, \vik{u}{i}{k-1} \}$, respectively.

%%%%%%%%%%%%%%%%%%%%%%%%%%%%%%%%%%%%%%%%%%%%%%%%%%%%%%%%%%%%%%%%%%%%%%%%%%%%%%%%
\subsection{State estimation}
\label{sec:stateEstimation}
%\Seb{Consider removing indices i again in this section?}
The local state estimator on agent $i$ has access to all measurements $\Yall_k^i$ and inputs $\Uall_k^i$ (\cf \fig
\ref{fig:remoteEstimation}). The Kalman filter (KF)
is the optimal Bayesian estimator in this setting, \cite{AnMo05}; it recursively computes the exact posterior PDF
% cut (space): of $x_k$ conditioned on $\Yall_k$, which we denote by  
$f(\xik{i}{k}| \Yall^i_k,\Uall^i_k)$.
% by  Indeed, we have $f(x_k | \Yall_{k-1}) = \Nc(x_k; \xKF_{k|k-1}, \PKF_{k|k-1})$ and $f(x_k | \Yall_{k}) = \Nc(x_k; \xKF_{k|k}, \PKF_{k|k})$, where the conditional mean and variances are computed recursively by
The KF recursion is
% given by
\begin{align}
\xKFik{i}{k|k-1} &= A_{i} \xKFik{i}{k-1} + B_{i} \vik{u}{i}{k-1} \label{eq:KF1} \\
\PKFik{i}{k|k-1} &= A_i \PKFik{i}{k-1} A_{i}^\transp + Q_{i} =: \Voi{i}(\PKFik{i}{k-1}) \label{eq:KF2} \\
L_k^i &= \PKFik{i}{k|k-1} H_i^\transp (H_i \PKFik{i}{k|k-1} H_i^\transp + R_i)^{-1} \label{eq:KF3_Lk} \\
\xKFik{i}{k} &= \xKFik{i}{k|k-1} + L_k^i(\vik{y}{i}{k} - H_i \xKFik{i}{k|k-1}) \label{eq:KF4} \\
\PKFik{i}{k} &= (I -L_k^i H_i) \PKFik{i}{k|k-1}.  \label{eq:KF5}
% omitted: =: \Vck{k}(\PKF_{k|k-1})
\end{align}
% With indices i everywhere
%\begin{align}
%\xKFik{i}{k|k-1} &= A_{i} \xKFik{i}{k-1} + B_{i} \vik{u}{i}{k-1} \label{eq:KF1} \\
%\PKFik{i}{k|k-1} &= A_i \PKFik{i}{k-1} A_{i}^\transp + Q_{i} =: \Voi{i}(\PKFik{i}{k-1}) \label{eq:KF2} \\
%L_k^i &= \PKFik{i}{k|k-1} H_i^\transp (H_i \PKFik{i}{k|k-1} H_i^\transp + R_i)^{-1} \label{eq:KF3_Lk} \\
%\xKFik{i}{k} &= \xKFik{i}{k|k-1} + L_k^i(\vik{y}{i}{k} - H_i \xKFik{i}{k|k-1}) \label{eq:KF4} \\
%\PKFik{i}{k} &= (I -L_k^i H_i) \PKFik{i}{k|k-1}.  \label{eq:KF5}
%% omitted: =: \Vck{k}(\PKF_{k|k-1})
%\end{align}
%
% Previous notation:
%\begin{align}
%\xKF_{k|k-1} &= A_{k-1} \xKF_{k-1} + B_{k-1} u_{k-1} \label{eq:KF1} \\
%\PKF_{k|k-1} &= A_{k-1} \PKF_{k-1} A_{k-1}^\transp + Q_{k-1} =: \Vok{k-1}(\PKF_{k-1}) \label{eq:KF2} \\
%L_k &= \PKF_{k|k-1} H_k^\transp (H_k \PKF_{k|k-1} H_k^\transp + R_k)^{-1} \label{eq:KF3_Lk} \\
%\xKF_{k} &= \xKF_{k|k-1} + L_k(y_k - H_k \xKF_{k|k-1}) \label{eq:KF4} \\
%\PKF_{k} &= (I -L_k H_k) \PKF_{k|k-1}.  \label{eq:KF5}
%% omitted: =: \Vck{k}(\PKF_{k|k-1})
%\end{align}
where $f(\xik{i}{k} | \Yall_{k-1}^i, \Uall_k^i) = \Nc(\xik{i}{k}; \xKFik{i}{k|k-1}, \PKFik{i}{k|k-1})$, $f(\xik{i}{k} | \Yall_{k}^i, \Uall_k^i)$ $= \Nc(\xik{i}{k}; \xKFik{i}{k}, \PKFik{i}{k})$, and the short-hand notation $\xKFik{i}{k} = \xKFik{i}{k|k}$ and $\PKFik{i}{k} = \PKFik{i}{k|k}$ is used for the posterior variables.
%The superscript `F' is used to denote the KF with \emph{full data} in distinction to the later event-based estimator.  
In \eqref{eq:KF2}, we introduced the short-hand notation $\Voi{i}$ for the open-loop variance update for later reference.  
%In \eqref{eq:KF2} and \eqref{eq:KF5}, we introduced the short-hands $\Vok{k-1}$ and $\Vck{k}$ for open-loop and closed-loop variance updates for later reference.  
%Also, we sometimes use short-hand notation for the posterior variables, $\xKF_{k-1} := \xKF_{k-1|k-1}$ and $\PKF_{k-1} := \PKF_{k-1|k-1}$.  
We shall also need the $M$-step ahead prediction of the state ($M \geq 0$),
%; that is, $(x_{k+M} | \Yall_k)$, 
whose PDF is given by \cite[p.~111]{AnMo05}
\begin{equation}
f(\xik{i}{k+M} | \Yall_k^i, \Uall_{k+M}^i) = \Nc(\xik{i}{k+M}; \, \xKFik{i}{k+M|k}, \PKFik{i}{k+M|k}),
\label{eq:PDFstatePredM}
\end{equation}
with mean and variance obtained by the open-loop KF iterations \eqref{eq:KF1}, \eqref{eq:KF2}, \ie
$\xKFik{i}{k+M|k} = A_i^{M}\xKFik{i}{k} + \sum_{m=1}^M A_i^{M-m}B \uik{i}{k+m-1}$ 
and $\PKFik{i}{k+M|k} = (\Voi{i} \circ \cdots \circ
\Voi{i}) (\PKFik{i}{k})$,
%\begin{align}
%\xKF_{k+M|k} &= \Phi_{(k+M-1):k} \xKF_k \\ %\label{eq:KF_meanPred} \\
%\PKF_{k+M|k} &= (\Vok{k+M-1} \circ \cdots \circ \Vok{k+1} \circ \Vok{k}) (\PKF_k)  %\label{eq:KF_varPred}
%\end{align}
where `$\circ$' denotes composition.  Finally, the error of the KF is defined as 
\begin{equation}
\vik{\eKF}{i}{k} := \vik{x}{i}{k} - \xKFik{i}{k} .
\label{eq:def_KFerror}
\end{equation}

%%%%%%%%%%%%%%%%%%%%%%%%%%%%%%%%%%%%%%%%%%%%%%%%%%%%%%%%%%%%%%%%%%%%%%%%%%%%%%%%
\subsection{Control}
\label{sec:control}
Because we are considering coordination of multiple things, the $i$'s control input may depend on the prediction of the other things in the IoT (\cf \fig \ref{fig:agentSchematic}).  We thus consider a control policy
\begin{equation}
\uik{i}{k-1} = F_i \xKFik{i}{k-1} + \! \sum_{j \in \N_N \setminus \{i\}}  \!  F_j \xPik{j}{k-1}
\label{eq:controlLawPre}
\end{equation}
where the local KF estimate $\xKFik{i}{k}$ is combined with predictions
$\xPik{j}{k}$ of the other agents (to be made precise below), and $\N_N$ denotes the set of all integers $\{1, \dots, N\}$.  
For coordination schemes where not all agents need to be coupled, some $F_j$ may be zero.  Then, these states do not need to be predicted.

It will be convenient to introduce the auxiliary variable $\xi_{k}^i = \sum_{j \in \N_N \setminus \{i\}}  \!  F_j \xPik{j}{k}$; \eqref{eq:controlLawPre} thus becomes 
\begin{equation}
\uik{i}{k-1} = F_i \xKFik{i}{k-1} + \xi_{k-1}^i.
\label{eq:controlLaw}
\end{equation}

%%%%%%%%%%%%%%%%%%%%%%%%%%%%%%%%%%%%%%%%%%%%%%%%%%%%%%%%%%%%%%%%%%%%%%%%%%%%%%%%
\subsection{Communication network}
\label{sec:commNetwork}
%Agent $i$ will decide when to broadcast its local state $\xKFik{i}{k}$ over the network.
Communication between agents occurs over a bus network that connects \emph{all} things with each other.  In particular, we assume that data (if transmitted) can be received by all agents that care about state information from the sending agent:
% in the IoT:
\begin{assumption}
Data transmitted by one agent can be received by all other agents in the IoT.
\end{assumption}
Such bus-like networks are common, for example, in automation industry in form of \emph{wired} fieldbus systems \cite{Th05}, but have recently also been proposed for low-power multi-hop \emph{wireless} networks \cite{FerZimLWB,mager2018feedback}.
%
%For the purpose of developing the triggers, we abstract communication to be ideal, without delay and packet loss. 
For the purpose of developing the triggers, we further abstract communication to be ideal:
\begin{assumption}
Communication between agents is without delay and packet loss.
\end{assumption}
This assumption is dropped later in the multi-vehicle simulation.

%%%%%%%%%%%%%%%%%%%%%%%%%%%%%%%%%%%%%%%%%%%%%%%%%%%%%%%%%%%%%%%%%%%%%%%%%%%%%%%%
\subsection{State prediction}
The sensor agent in \fig \ref{fig:remoteEstimation} sporadically communicates its local estimate $\xKFik{i}{k}$
to the remote estimator, which, at every step $k$, computes its own state estimate $\xPik{i}{k}$ from the available data via state prediction.
%
%We consider an event-based architecture, where the sensor sporadically communicates its local estimate $\xKF_k$ to the remote estimator, which, at every step $k$, computes its own state estimate $\hat{x}_k$ from the available data.  Other event-based architectures  are also conceivable, for example, where measurements $y_k$ instead of state estimates are communicated as in \cite{TrDAn11,TrDAn14b,TrCa15}, which can be beneficial for practical considerations (\eg when $n_\text{y} \ll n_\text{x}$) or in distributed architectures.
%
We denote by $\gamma_k^i \in \{0, 1\}$ the decision taken by the sensor about whether an update is sent ($\gamma_k^i = 1$) or not ($\gamma_k^i = 0$).  For later reference, we denote the set of all triggering decisions until $k$ by $\Gamall_k^i := \{ \gamma_1^i, \gamma_2^i, \dots, \gamma_k^i \}$.
% CUT:
%\begin{equation}
%\Gamall_k := \{ \gamma_1, \gamma_2, \dots, \gamma_k \}.
%\label{eq:Gamall}
%\end{equation}
%
%For the purpose of developing the triggers, we abstract communication to be ideal, without delay and with zero probability of packet loss. This assumption is dropped later in the multi-vehicle simulation.

The state predictor on the remote agent (\cf \fig \ref{fig:remoteEstimation}) uses the following recursion to compute $\xPik{i}{k}$, its remote estimate of $\xik{i}{k}$:
\begin{align}
\xPik{i}{k} &= 
\begin{cases}
A_i \xPik{i}{k-1}  + B_{i} \check{u}_{k-1} & \text{if $\gamma_k^i =
0$}
\\
\xKFik{i}{k}  & \text{if $\gamma_k^i = 1$} ;
\end{cases}
\label{eq:remoteEst_pre}
\end{align}
that is, at times when no update is received from the sensor, the estimator predicts its previous estimate according to the process model \eqref{eq:sys_x} and prediction of the input \eqref{eq:controlLaw} by
\begin{equation}
\check{u}^{i}_{k-1} = F_i \xPik{i}{k-1} + \xi_{k-1}^i .
%\check{u}^{i}_{k-1} = F_i \xPik{i}{k-1} + \! \sum_{m \in \{1,\dots,N\} \setminus i}  \!  F_m \xPik{m}{k-1} .
%= \sum_{m \in \{1,\dots,N\} }  \!  F_m \xPik{m}{k-1}
\label{eq:controlLawPred}
\end{equation}
Implementing \eqref{eq:controlLawPred} thus requires the remote agent to run predictions of the form \eqref{eq:remoteEst_pre} for all other things $m$ that are relevant for computing $\xi_{k-1}^i$.
%the control input $\check{u}^{i}_{k-1}$.  
This is feasible as an agent can broadcast state updates (for $\gamma_k^i=1$) to all other things via the bus network.
We emphasize that $\xi_{k-1}^i$, the part of the input $u^i_{k-1}$ that corresponds to all other agents, is known exactly on the remote estimator, since updates are sent to all agents connected to the network synchronously.  Hence, the difference between the actual input \eqref{eq:controlLaw} and predicted input \eqref{eq:controlLawPred}
%, $u^i_{k-1}$ and $\check{u}^i_{k-1}$, 
stems from a difference in $\xKF^i_{k-1}$ and $\xP^i_{k-1}$.

With \eqref{eq:controlLawPred}, the prediction \eqref{eq:remoteEst_pre} then becomes 
\begin{align}
\xPik{i}{k} &= 
\begin{cases}
\bar{A}_i \xPik{i}{k-1}  + B_i\xi_{k-1}^i 
%=: \xI_k 
& \text{if $\gamma_k^i = 0$} \\
\xKFik{i}{k} 
%=: \xII_k 
& \text{if $\gamma_k^i = 1$} ;
\end{cases}
\label{eq:remoteEst}
\end{align}
where $\bar{A}_i := A_i + B_i F_i$ denotes the closed-loop state transition matrix of agent $i$.
%For easier reference and distinguishing the two paths that the state predictor \eqref{eq:remoteEst} can take, also we introduced the variables $\xI$ and $\xII$, corresponding to the case of communication or not, respectively.
The estimation error at the remote agent, we denote by 
%$\vik{e}{i}{k} := \xik{i}{k} - \xPik{i}{k}$.
\begin{equation}
\vik{e}{i}{k} := \xik{i}{k} - \xPik{i}{k}.
\label{eq:def_remoteError}
\end{equation}

% Errors: introduced later
%Furthermore, we introduce the corresponding errors, $\eI_k := x_k - \xI_k$ and $\eII_k := x_k - \xII_k$.
%The general estimation error, we denote by $e_k := x_k - \hat{x}_k$.

A copy of the state predictor \eqref{eq:remoteEst} is also implemented on the sensor agent to be used for the triggering decision (\cf \fig \ref{fig:remoteEstimation}).

%The remote estimator thus corresponds to the open-loop prediction of the KF according to \eqref{eq:PDFstatePredM} (mean).
%Indeed, let $\last_k \leq k$ denote the last time that data was transmitted; then $\hat{x}_k = \xKF_{k|\last_k}$.  

Finally, we comment how local estimation quality can possibly be further improved in certain applications.
\begin{remark}
\label{rem:localSensors}
In \eqref{eq:remoteEst}, agent $j$ makes a pure state prediction about agent $i$'s state in case of no communication from agent $i$ ($\gamma_k^i=0$).  If agent $j$ has additional local sensor information about agent $i$'s state, it may employ this by combining the prediction step with a corresponding measurement update.  This may help to improve estimation quality (\eg obtain a lower error variance).
% in some application, we do not explicitly consider this case herein.  
In such a setting, the triggers developed herein can be interpreted as `conservative' triggers that take only prediction into account.
\end{remark}

% Previously cut (for space):
\begin{remark}
\label{rem:negativeInformation}
Under the assumption of perfect communication, the event of not receiving an update ($\gamma_k^i=0$) may also contain information useful for state estimation (also known as \emph{negative information} \cite{SiNoHa13}).  Here, we disregard this information in the interest of a straightforward estimator implementation (see \cite{TrCa15} for a more detailed discussion).  
\end{remark}

%%%%%%%%%%%%%%%%%%%%%%%%%%%%%%%%%%%%%%%%%%%%%%%%%%%%%%%%%%%%%%%%%%%%%%%%%%%%%%%%
\subsection{Problem formulation}
\label{sec:objective}
The main objective of this article is the development 
of principled ways for predicting future triggering decisions. In particular, we shall develop two concepts:
\begin{enumerate}
\item \emph{predictive triggering:} at every step $k$ and for a fixed horizon $M\!>\!0$, $\gamma_{k+M}^i$ is predicted, \ie whether or not communication is needed at $M$ steps in future; and
\item \emph{self triggering:} the next trigger is predicted at the time of the last trigger.
%\item at every step $k$, predicting $\gamma_{k+M}$, \ie whether communication is needed $M>0$ steps in future (\emph{predictive triggering}) 
%\item at the time of the last trigger, predicting the next one (\emph{self triggering}) 
\end{enumerate}

In the next sections, we develop these triggers for agent $i$ shown in \fig \ref{fig:remoteEstimation}, which is representative for any one agent in \fig \ref{fig:iotControlSchematic}.
%considering the setup in \fig \ref{fig:agentSchematic} and as described above.  
Because we will thus discuss estimation, triggering, and prediction solely for agent $i$ (\cf \fig \ref{fig:remoteEstimation}), we drop the index `$i$' to simplify notation.  Agent indices are re-introduced in \sect \ref{sec:multiAgent}, when again multiple agents are considered.

%In the following sections, we will develop these triggers for one agent $i$ considering the setup in \fig \ref{fig:agentSchematic} and as described above.  To simplify notation until further notice, we drop the index `$i$' until otherwise noted.

For ease of reference, key variables from this and later sections are summarized in Table~\ref{tab:summary}.
%\Dom{I included the arraystretch command with a value of 1.2 in the table. Reason is that the superscript i's were ``eaten'' by the indices k at some places. This enlarges the vertical spacing, thus the table needs a bit more space, but I think it looks better.}
% ST: Good!
\begin{table}[tb]
\centering
\caption{Summary of main variables used in the article.  The agent index `$i$' is dropped for all variables in \sect \ref{sec:triggeringFramework} to \ref{sec:analysisTriggers}.}
{\renewcommand{\arraystretch}{1.2}
\begin{tabular}{llll}
$A_i, B_i, H_i, Q_i, R_i$
%$A_i$, $B_i$, $H_i$, $Q_i$, $R_i$ 
& Dynamic system parameters\\
$F_i$ & Control gain corresponding to agent $i$'s state \\
$x_k^i$ &  State of agent $i$, eq.\ \eqref{eq:sys_x} %($i$ omitted in \sect \ref{sec:triggeringFramework} to \ref{sec:analysisTriggers}) 
\\
$\xKF_k^i$ & Kalman filter (KF) estimate \eqref{eq:KF4} \\
$\xPik{i}{k}$ & Remote state estimate \eqref{eq:remoteEst} \\
%\rule{0pt}{2.2ex}$\xKF_k^i$ & Kalman filter (KF) estimate \\
%\rule{0pt}{2.2ex}$\xI_k$ & State prediction (case of no communication) \\ 
%\rule{0pt}{2.2ex}$\xII_k=\xKF_k$ & State prediction reset (case of
%communication) \\
$\vik{\eKF}{i}{k}$ & KF estimation error \eqref{eq:def_KFerror} \\
$\vik{e}{i}{k}$ & Remote estimation error \eqref{eq:def_remoteError} \\
$\gamma_k^i$ & Communication decision (1=communicate, 0=not) \\
$\Gamma_k^i$ & Set of communication decisions $\{\gamma^i_1, \dots, \gamma^i_k \}$ \\
$X\notrigk$, $X\trigk$ & Expression $X$ evaluated for resp.\ $\gamma_k=0$, $\gamma_k=1$ \\
$\Yall_k^i$ & Set of all measurements on agent $i$ until time $k$ \\
$\Uall_k^i$ & Set of all inputs on agent $i$ until time $k$ \\
$\tilde{x}_k$, $\tilde{e}_k$, etc. & Collection of corresponding variables for all agents \\
%$\xKF=\xII_k$ & Closed loop Kalman filter state\\
%$\Vok{k-1}$ & Open loop variance update\\
%$\Phi_{k_2:k_1}$ & Successive application of $A_k$\\
$C_k$ & Communication cost (`$i$' dropped) \\ 
$E_k$ & Estimation cost (`$i$' dropped) \\
$M$ & Prediction horizon (`$i$' dropped) \\
$\last_k$ & Last triggering time (`$i$' dropped) \\
$\lastel_k$ & Time of last nonzero elem.\ in $\Gamma_{k+M}$ (`$i$' dropped) \\
$\Delta$ & Number of steps from $\lastel_{k-1}$ to $k\!+\!M$ (\cf Lem.\ \ref{lem:PDF_eI}) \\
%, $\Delta := k+M-\kappa_{k-1}$
$\N_N$ & Set of integers $\{1, \dots, N\}$ \\
$\mathbb{E}[X_1|X_2]$ & Expected value of $X_1$ conditioned on $X_2$\\
$f(X_1|X_2)$ & Probability density fcn (PDF) of $X_1$ cond.\ on $X_2$
\end{tabular}
}
\label{tab:summary}
\end{table}

%% file: content/4_TriggeringFramework.tex
\section{Triggering Framework}
\label{sec:triggeringFramework}
To develop a framework for making predictive triggering decisions, we extend the approach from \cite{TrCa15}, where triggering is formulated as a one-step optimal decision problem trading off estimation and communication cost. 
While this framework was used in \cite{TrCa15}  to re-derive existing event triggers
% (summarized in \sect \ref{sec:eventTrigger}), as well as re-derive existing ones 
% in a unified way 
 (summarized in \sect \ref{sec:eventTrigger}), we extend the framework herein to yield predictive and self triggering (\sect \ref{sec:predTrigger} and \ref{sec:selfTrigger}).

%%%%%%%%%%%%%%%%%%%%%%%%%%%%%%%%%%%%%%%%%%%%%%%%%%%%%%%%%%%%%%%%%%%%%%%%%%%%%%%%
\subsection{Decision framework for event triggering}
\label{sec:eventTrigger}
The sensor agent (\cf \fig \ref{fig:remoteEstimation}) makes a decision 
%at time $k$ 
between using the communication channel (and thus paying a communication cost $C_k$) to improve the remote estimate, or to save communication, but pay a price in terms of a deteriorated estimation performance (captured by a suitable estimation cost $E_k$).  The communication cost $C_k$ is application specific and may be associated with the use of bandwidth or energy, for example.  We assume $C_k$ is known for all times $k$.  The estimation cost $E_k$ is used to measure the discrepancy between the remote estimation error $e_k$ without update ($\gamma_k=0$), which we write as $e_k\notrigk$, and with update, $e_k\trigk$.
%$\eI_k:= x_k - \xI_k$ and with update $\eII_k := x_k - \xII_k$; 
%\ie $E_k = E(\eI_k, \eII_k)$
%for a suitable choice of $E$;
%for example,
Here, we choose
\begin{equation}
E_k = e_k^\transp e_k \notrigk - e_k^\transp e_k \trigk
%E_k = (\eI_k)^\transp \eI_k - (\eII_k)^\transp \eII_k 
\label{eq:Ek_squares}
\end{equation}
comparing the difference in quadratic errors.
%: error $\eII$ will generally be ``smaller'' than $\eI$ because it has more information, thus, the discrepancy between the two captures that cost of estimation that is paid when not communicating.  
%For example,  
%\begin{equation}
%E_k = (\eI_k)^2 - (\eII_k)^2
%\label{eq:Ek_squares_scalar}
%\end{equation}
%was used in \cite{TrCa15} for scalar quantities. This cost measures in terms of quadratic errors how much worse the error without update ($\eI_k$) is, compared to the one  with update ($\eII_k$).

Formally, the triggering decision can then be written as
\begin{equation}
%\text{at time $k$:} \quad
\min_{\gamma_k \in \{0, 1\}}  \gamma_k \commC_k + (1-\gamma_k) E_k .
\label{eq:optProblET_ideal}
\end{equation}
Ideally, one would like to know $e_k\notrigk$ and $e_k\trigk$ exactly when computing the estimation cost in order to determine whether it is worth paying the cost for communication. However, 
%$\eI_k$ and $\eII_k$ 
$e_k$ cannot be computed since the true state is generally unknown (otherwise we would not have to bother with state estimation in the first place).
%(see \eqref{eq:eI}, \eqref{eq:eII}).
As is proposed in \cite{TrCa15}, we consider instead the expectation of $E_k$ conditioned on the data $\Dc_k$ that is available by the decision making agent.  Formally, 
\begin{equation}
%\text{at time $k$:} \quad
\min_{\gamma_k \in \{0, 1\}}  \gamma_k \commC_k + (1-\gamma_k) \, \E[ E_k | \Dc_k ] 
%\min_{\gamma_k \in \{0, 1\}}  \gamma_k \commC_k + (1-\gamma_k) \bar{E}_k
\label{eq:optProblET}
\end{equation}
%with $\bar{E}_k = \E[ E_k | \Dc_k ] $.  
which directly yields the triggering law
\begin{equation}
\text{at time $k$:} \quad \gamma_k = 1  \; \Leftrightarrow \; \E[ E_k | \Dc_k ] \geq C_k .
\label{eq:ETgeneral}
\end{equation}
In \cite{TrCa15}, this framework was used to re-derive common event-triggering mechanisms such as innovation-based triggers \cite{TrDAn11,WuJiJoSh13}, 
% removed: Tr12, Tr17
or variance-based triggers \cite{TrDAn14b,LeDeQu15}, depending on whether the current measurement $y_k$ is included in $\Dc_k$, or not. 
% removed (space, not directly relevant): SiKeNo14

\begin{remark}
The choice of quadratic errors in \eqref{eq:Ek_squares} is only one possibility for measuring the discrepancy between $e_k\notrigk$ and $e_k\trigk$ and quantifying estimation cost. It is motivated from the objective of keeping the squared estimation error small, a common objective in estimation.  The estimation cost in \eqref{eq:Ek_squares} is positive if the squared error $e_k^\transp e_k\notrigk$ (\ie without communication) is larger than $e_k^\transp e_k\trigk$ (with communication), which is to be expected on average.  Moreover, the quadratic error is convenient for the following mathematical analysis.
Finally, the scalar version of \eqref{eq:Ek_squares} was shown in \cite{TrCa15} to yield common known event triggers.
However, other choices than \eqref{eq:Ek_squares} are clearly conceivable, and the subsequent framework can be applied analogously. 
\end{remark}

% When solving  $\bar{E}_k := \E[ E_k | \Dc_k ] $, \eqref{eq:optProblET} then directly yields the \emph{event trigger}
%\begin{equation}
%\text{at time $k$:} \quad \gamma_k = 1  \; \Leftrightarrow \; \E[ E_k | \Dc_k ] \geq C_k ,
%\label{eq:ETgeneral}
%\end{equation}
%which is to be evaluated at time $k$.
%\todo{Introduce $\trigsig_{k|k} = \E[ E_k | \Dc_k ]$.}

%%%%%%%%%%%%%%%%%%%%%%%%%%%%%%%%%%%%%%%%%%%%%%%%%%%%%%%%%%%%%%%%%%%%%%%%%%%%%%%%
\subsection{Predictive triggers}
\label{sec:predTrigger}
This framework can directly be extended to derive a predictive trigger as formulated in \sect \ref{sec:objective}, which makes a communication decision $M$ steps in advance, where $M\!>\!0$ is fixed by the designer.
%It is straightforward to 
%this framework can be extended to let the triggering agent make a communication decision $M$ steps in advance, where $M\!>\!0$ is fixed by the designer.  Namely, 
Hence, we consider the future decision on $\gamma_{k+M}$ and condition the future estimation cost $E_{k+M}$
%optimization problem \eqref{eq:optProblET_ideal}  
on $\Dc_k = \{ \Yall_k, \Uall_k \}$, the data available at the current time $k$.  Introducing $\trigsig_{k+M|k} := \E[ E_{k+M} | \Yall_k, \Uall_k ]$, the optimization problem \eqref{eq:optProblET_ideal} then becomes
%\footnote{The cost $C_k$ is assumed to be known in advance for all times $k$.}
\begin{equation}
%\text{at time $k$:} \quad
\min_{\gamma_{k+M} \in \{0, 1\}}  \gamma_{k+M} \commC_{k+M} + (1-\gamma_{k+M}) \trigsig_{k+M|k} 
\label{eq:optProblPT}
\end{equation}
which yields the \emph{predictive trigger} (PT):
\begin{equation}
\text{at time $k$:} \quad \gamma_{k+M} = 1  \; \Leftrightarrow \; \trigsig_{k+M|k} \geq C_{k+M} .
\label{eq:PTgeneral}
\end{equation}
%Here, we have $\Dc_k = \Yall_k \cup \Gamall_{k+M-1}$, \ie all measurements until $k$ and all triggering decisions that have been made until then.
% until $M$ steps ahead.
In \sect \ref{sec:triggers}, we solve $\trigsig_{k+M|k} = \E[ E_{k+M} | \Yall_k, \Uall_k ]$ for 
%a specific 
the choice of error  
%\eqref{eq:estCostFunction}
\eqref{eq:Ek_squares} 
to obtain an expression for the trigger \eqref{eq:PTgeneral} in terms of the problem parameters.

%%%%%%%%%%%%%%%%%%%%%%%%%%%%%%%%%%%%%%%%%%%%%%%%%%%%%%%%%%%%%%%%%%%%%%%%%%%%%%%%
\subsection{Self triggers}
\label{sec:selfTrigger}
A self trigger computes the next triggering instant at the time when an update is sent.
%at the time, when an update is sent, the next triggering instant is already computed.  
% CUT:
%This allows the sensor to go to sleep, for example, or communicate the next triggering instant alongside the current data to a network manager so as to reconfigure network resources accordingly.  
%
A self triggering law is thus obtained 
% for estimation can be obtained 
 by solving \eqref{eq:PTgeneral} at time $k = \last_k$ for the smallest $M$ such that $\gamma_{k+M} = 1$.  
Here, $\last_k \leq k$ denotes the last triggering time; in the following, we drop `$k$' when clear from context and simply write $\last_k = \last$.  Formally, the \emph{self trigger} (ST) is then given by:
%Then, a \emph{self-trigger} determining at $k=\last$ the next communication instant $\gamma_{k+M} = 1$ is defined by
%\begin{align}
%&\text{trigger at} \,\, k=\last + M \,\, \text{with}  \,\, M \geq 1 \,\, \text{such that} \\
%%\text{find} \,\, &M \geq 1 \,\, \text{such that} \\
%&\quad \E[ E_{\last+1} | \Yall_\last ] < C_{\last+1}, \, \dots , \,
%\E[ E_{\last+M-1} | \Yall_\last ] < C_{\last+M-1}, \nonumber \\
%&\quad \text{and} \,\,\, \E[ E_{\last+M} | \Yall_\last ] \geq C_{\last+M} . \nonumber
%\end{align}
\begin{align}
\!\text{at time $k\!=\!\last$:} \,\,\, &\text{find smallest $M\! \geq\! 1$ s.t.\ $\trigsig_{\last+M|\last} \geq C_{\last+M}$}, \nonumber \\[-1mm]
& \text{set} \, \gamma_{\ell+1} \!=\! \dots  \!=\! \gamma_{\ell+M-1}\!=\!0, \gamma_{\ell+M}\!=\!1.
\label{eq:STgeneral}
\end{align}
%
% Initial version:
%\begin{align}
%\text{at time $k=\last$:} \quad &\text{find smallest $M \geq 1$ s.t.\ $\trigsig_{\last+M|\last} \geq C_{\last+M}$} .
%\label{eq:STgeneral}
%\end{align}
%with $\Dc_\last = \Yall_\last \cup \Gamall_\last$.  

While both the PT and the ST compute the next trigger ahead of time, they represent two different triggering concepts.
% which may be advantageous depending on the application.  
The PT \eqref{eq:PTgeneral} is evaluated at every time step $k$ with a given prediction horizon $M$, whereas the ST \eqref{eq:STgeneral} needs to be evaluated at $k = \last$ only and yields (potentially varying) $M$.  That is, $M$ is a \emph{fixed} design parameter for the PT, and \emph{computed} with the ST.
%In contrast to the PT \eqref{eq:PTgeneral}, which is evaluated at every $k$ using the most recent data $\Yall_k$, the ST needs to be evaluated at $k = \last$ only.  
Which of the two should be used depends on the application (\eg whether continuous monitoring of the error signal is desirable).  
The two types of triggers will be compared in simulations and experiments in subsequent sections.
%In \sect \ref{sec:simulations}, the two concepts are compared in terms of their effectiveness in trading off estimation quality and communication.

%% file: content/5_PTST.tex
\section{Predictive Trigger and Self Trigger}
\label{sec:predAndSelfTrig}
% other: Some Triggers
\label{sec:triggers}
Using the triggering framework of the previous section, we derive concrete instances of the self and predictive trigger for the squared estimation cost \eqref{eq:Ek_squares}.  
%
%Other choices for measuring the discrepancy between $\eI$ and $\eII$ are also conceivable, and the framework can be applied analogously.  The specification \eqref{eq:Ek_squares} is reasonable if keeping the squared estimation error $(e_k)^\transp e_k$ small is of interest, which is a typical objective in estimation.  The estimation cost in \eqref{eq:Ek_squares} is positive if the squared error $(\eI_k)^\transp \eI_k$ (\ie without communication) is larger than $(\eII_k)^\transp \eII_k$ (with communication), which is to be expected on average.  The scalar version of \eqref{eq:Ek_squares} was used in \cite{TrCa15} to derive optimal event triggers.
%
%Before computing the self and predictive trigger, 
To this end, we first determine the PDF of the estimation errors.
%For computing the triggers, the following lemmas are instrumental.

%%%%%%%%%%%%%%%%%%%%%%%%%%%%%%%%%%%%%%%%%%%%%%%%%%%%%%%%%%%%%%%%%%%%%%%%%%%%%%%%
\subsection{Error distributions} 	
\label{sec:errorDistributions}
%Choosing the function $E$ in \eqref{eq:estCostFunction} (\ie specifying how one measures the discrepancy between the estimation errors $\eI_k$ and $\eII_k$), the conditional expectations in \eqref{eq:PTgeneral} and \eqref{eq:STgeneral} can be solved to obtain predictive triggers and self-triggers, respectively.  
In this section, we compute the conditional error PDF $f(e_{k+M} | \Yall_{k}, \Uall_k)$ for the cases $\gamma_{k+M}=0$ and $\gamma_{k+M}=1$,
%conditional error PDFs $f(\eI_{k+M} | \Yall_{k}, \Uall_k)$ and $f(\eII_{k+M} | \Yall_{k}, \Uall_k)$,
 which characterize the distribution of the estimation cost $E_{k+M}$ in \eqref{eq:Ek_squares}.  These results are used in the next section to solve for the triggers \eqref{eq:PTgeneral} and \eqref{eq:STgeneral}.
 % for the specific choice \eqref{eq:Ek_squares} of $E$.

%The framework in the previous section can be used to derive event-, predictive, and self-triggers depending on how the estimation cost $E_k$ is defined; that is, how one measures the cost of having $\eI_k$, instead of $\eII_k$ at the remote estimator.  Since, in general, the estimation cost depends on $E_k$ on $\eI_k$ and $\eII_k$, we first characterize the conditional PDFs of the errors in this section, which is used in the next section to solve for the predictive trigger \eqref{eq:PTgeneral} and self-trigger \eqref{eq:STgeneral} under a specific choice for $E_k$.

Both triggers \eqref{eq:PTgeneral} and \eqref{eq:STgeneral} predict the communication decisions $M$ steps ahead of the current time $k$.
% ($M$ is a design parameter for \eqref{eq:PTgeneral} and computed in case of \eqref{eq:STgeneral}). 
Hence, in both cases, the set of triggering decisions $\Gamall_{k+M}$ can be computed from the data $\Yall_k$, $\Uall_k$.  In the following, it will be convenient to denote the time index of the last nonzero element in $\Gamall_{k+M}$ (i.e., the last planned triggering instant) by $\lastel_k$; for example, for $\Gamall_{10} = \{ \dots, \gamma_8 = 1, \gamma_9=1, \gamma_{10}=0 \}$, $k=6$, and $M=4$, we have $\lastel_{6} = 9$.  It follows that $\lastel_k \geq \last_k$, with equality $\lastel_k = \last_k$ if no trigger is planned for the next $M$ steps.  

% Omitted (space, not necessary?):
%
%Figure \ref{fig:illustrationTimeIndices} illustrates these definitions by means of an example.
%\begin{figure}[tb]
%\centering
%\includegraphics[scale=0.5]{dummy.eps}
%\caption{\todo{add example with time lines} \todo{omit?}}
%\label{fig:illustrationTimeIndices}
%\end{figure}

%For the following, it is convenient to write the current time $k$ as $k=\last + K$, where $\last$ is the last time when a communication was triggered (as previously defined) and $K$ are the steps elapsed since then.  Thus, we seek to characterize the predictive distributions $f(\eI_{k+M} | \Yall_{k}) =$ $f(\eI_{\last+K+M} | \Yall_{\last+K})$ and $f(\eII_{k+M} | \Yall_{k}) =$ $f(\eII_{\last+K+M} | \Yall_{\last+K})$.  These are relevant both for the predictive trigger \eqref{eq:PTgeneral} and the self-trigger \eqref{eq:STgeneral} (with $K=0$). 

The following two lemmas state the sought error PDFs.  
\begin{lemma}
\label{lem:PDF_eII}
For $\gamma_{k+M}=1$, the predicted error $e_{k+M}$ conditioned on $\Yall_k$, $\Uall_k$ is normally distributed with\footnote{The superscripts `c' and `nc' denote the cases `communication' ($\gamma=1$) and `no communication' ($\gamma=0$). \label{fn:notation_c_nc}}
\begin{align}
f(e_{k+M} | \Yall_{k}, \Uall_k) 
&= \Nc(e_{k+M}; \, \eIIhat_{k+M|k}, \PII_{k+M|k} ) \nonumber \\
&= \Nc(e_{k+M}; \, 0, \PKF_{k+M} ) . \label{eq:lem2_PDF_eII}
\end{align}
%
% Previous statement:
%The predicted error $\eII_{k+M}$ conditioned on $\Yall_k$, $\Uall_k$ is normally distributed with
%\begin{align}
%f(\eII_{k+M} | \Yall_{k}, \Uall_k) 
%&= \Nc(\eII_{k+M}; \, \eIIhat_{k+M|k}, \PII_{k+M|k} ) \nonumber \\
%&= \Nc(\eII_{k+M}; \, 0, \PKF_{k+M} ) . \label{eq:lem2_PDF_eII}
%\end{align}
\end{lemma}
\begin{proof}
See Appendix \ref{app:proofeII}.
\end{proof}

\begin{lemma}
\label{lem:PDF_eI}
For $\gamma_{k+M}=0$, the predicted error $e_{k+M}$ conditioned on $\Yall_k$, $\Uall_k$ is normally distributed\footnotemark[1] 
% Should be, but doesn't work: \footnotemark[\ref{fn:notation_c_nc}] 
\begin{equation}
f(e_{k+M} | \Yall_{k}, \Uall_k) = 
\Nc(e_{k+M}; \, \eIhat_{k+M|k}, \PI_{k+M|k} ) \label{eq:lem1_eIpdf}
\end{equation}
with mean and variance given as follows.
%
% Previous version:
%The predicted error $\eI_{k+M}$ conditioned on $\Yall_k$, $\Uall_k$ is normally distributed,
%\begin{equation}
%f(\eI_{k+M} | \Yall_{k}, \Uall_k) = 
%\Nc(\eI_{k+M}; \, \eIhat_{k+M|k}, \PI_{k+M|k} ) \label{eq:lem1_eIpdf}
%\end{equation}
%with mean and variance given as follows.
%
%
% Shorter version:
%
%with mean and variance given by, for $k > \lastel_{k-1}$,
%$\eIhat_{k+M|k} = \Phi_{(k+M-1):k} \, (\xKF_k - \xKF_{k|\last} )$ and
%$\PI_{k+M|k} = \PKF_{k+M|k}$;
%and, for $k \leq \lastel_{k-1}$, 
%$\eIhat_{k+M|k} = 0$  and
%$\PI_{k+M|k} = \PKF_{\lastel+\Delta|\lastel} = \PKF_{k+M|\lastel}$,
%
% Longer (nicer) version:

\underline{Case (i):} $k > \lastel_{k-1}$ (\ie no trigger planned within prediction horizon)
\begin{align}
\eIhat_{k+M|k} &= \bar{A}^M \, \Big(\xKF_k - \bar{A}^{k-\last} \xKF_{\last} - \sum_{m=1}^{k-\last} \bar{A}^{k-\ell-m} B \xi_{\ell+m-1}  \Big) \label{eq:lem1_eImean} \\
\PI_{k+M|k} &= \PKF_{k+M|k} + \Xi_{k,M} \label{eq:lem1_eIvar} 
\end{align}
where
\begin{align}
\Xi_{k,M} &:= \sum_{m=1}^{M-1} G_{M-m-1} L_{k+m} \tilde{P}_{k+m} L_{k+m}^\transp G_{M-m-1}^\transp , \\
\tilde{P}_{k}&:= HA \PKF_{k-1} A^\transp H^\transp + H Q H^\transp + R , \label{eq:Ptilde_def} \\
G_m&:= A G_{m-1} + BF \bar{A}^m, \quad G_0 := BF , \label{eq:G_def}
\end{align}
$L_k$ is the KF gain \eqref{eq:KF3_Lk}, and $\PKF_{k+M|k}$ is the KF prediction variance in \eqref{eq:PDFstatePredM}.
%\begin{align}
%\eIhat_{k+M|k} &= \Phi_{(k+M-1):k} \, (\xKF_k - \xKF_{k|\last} ) \label{eq:lem1_eImean} \\
%\PI_{k+M|k} &= \PKF_{k+M|k}  \label{eq:lem1_eIvar} 
%\end{align}
% was introduced.
%and $\Phi_{k_2:k_1}$ and $\Vok{k_1}$ are as defined in \eqref{eq:transitionPhi} and \eqref{eq:KF2}, and $\lastel$ is used as shorthand for $\lastel_{k-1}$.

\underline{Case (ii):} $k \leq \lastel_{k-1}$ (\ie trigger planned within horizon)
\begin{align}
\eIhat_{k+M|k} &= 0  \label{eq:lem1_eImean_b} \\
\PI_{k+M|k} &= \PKF_{\lastel+\Delta|\lastel} + \Xi_{\lastel,\Delta}
%= \PKF_{k+M|\lastel}  
\label{eq:lem1_eIvar_b}
\end{align}
where 
%$\Phi$ is as defined in \eqref{eq:transitionPhi}, 
$\lastel$ is used as shorthand for $\lastel_{k-1}$, and $\Delta := k+M-\kappa_{k-1}$.
%, and $\PKF_{\lastel+\Delta|\lastel}$ is the KF prediction variance in \eqref{eq:PDFstatePredM}.
\end{lemma}
\begin{proof}
See Appendix \ref{app:proofeI}.
\end{proof}

% Proofs moved to appendix
%
%We first prove \Lem \ref{lem:PDF_eII}, which will be used in the proof of \Lem \ref{lem:PDF_eI}.
%
%\begin{proof} {\it(Lemma \ref{lem:PDF_eII})}
%\input{content/Proof2.tex}
%\end{proof}
%
%\begin{proof} {\it(Lemma \ref{lem:PDF_eI})}
%\input{content/Proof1.tex}
%\end{proof}

A simpler formula for \Lem \ref{lem:PDF_eI} can be given for the case of an autonomous system \eqref{eq:sys_x} without input: 
%(\ie $B_i \uik{i}{k-1} = 0$):  
\begin{corollary}
For \eqref{eq:sys_x} with $B_i\uik{i}{k-1} = 0$, \eqref{eq:lem1_eIpdf} holds for case (i) with
% \eqref{eq:lem1_eImean_b} and \eqref{eq:lem1_eIvar_b} for $k \leq \lastel_{k-1}$ (unchanged), and, for $k > \lastel_{k-1}$:
\begin{align}
\eIhat_{k+M|k} &= A^M \, (\xKF_k - A^{k-\last} \xKF_{\last} ) \label{eq:lem1_eImean_noInput} \\
\PI_{k+M|k} &= \PKF_{k+M|k} \label{eq:lem1_eIvar_noInput} 
\end{align}
and for case (ii) with
\begin{align}
\eIhat_{k+M|k} &= 0  \label{eq:lem1_eImean_b_noInput} \\
\PI_{k+M|k} &= \PKF_{\lastel+\Delta|\lastel} .
\label{eq:lem1_eIvar_b_noInput}
\end{align}
\end{corollary}
\begin{proof}
Taking $B=0$ yields $\bar{A}=A$ and $\Xi_{k,M} = 0$ and thus the result.
\end{proof}
We thus conclude that the extra term $\Xi_{k,M}$ in the variance \eqref{eq:lem1_eIvar} stems from additional uncertainty about not exactly knowing future inputs.  
%This disinction is irrelevant in case (i) as a future communication is already schedule, leading to a reset to the KF estimate (which includes full information).

%%%%%%%%%%%%%%%%%%%%%%%%%%%%%%%%%%%%%%%%%%%%%%%%%%%%%%%%%%%%%%%%%%%%%%%%%%%%%%%%
\subsection{Self trigger}
The ST law \eqref{eq:STgeneral} is stated for a general estimation error $\trigsig_{\last+M|\last}$.  With the preceding lemmas, we can now solve for $\trigsig_{\last+M|\last}$ and obtain the concrete self triggering rule for the quadratic error \eqref{eq:Ek_squares}.
%Using the previous results, we solve \eqref{eq:STgeneral} with \eqref{eq:Ek_squares} to obtain a self triggering rule.
%To obtain a self trigger corresponding to the error definition \eqref{eq:Ek_squares}, we solve \eqref{eq:STgeneral} using the results from the previous section.  
\begin{proposition}
For the quadratic error \eqref{eq:Ek_squares}, the self trigger (ST) \eqref{eq:STgeneral} becomes: 
\begin{align}
&\text{find smallest $M \geq 1$ s.t.} \nonumber \\[-1mm]
&\text{$\trace( \PKF_{\last+M|\last} + \Xi_{\last,M} - \PKF_{\last+M} ) \geq C_{\last+M}$}; \nonumber \\[-1mm]
& \text{set} \,\,\, \gamma_{\ell+1} \!=\! \dots  \!=\! \gamma_{\ell+M-1}\!=\!0, \gamma_{\ell+M}\!=\!1.
\label{eq:STsquaredError}
\end{align}
\end{proposition}
\begin{proof}
Applying \Lem \ref{lem:PDF_eII} and \Lem \ref{lem:PDF_eI} (for $k = \last = \lastel_{k-1}$), we obtain
\begin{align}
&\trigsig_{\last+M |\last} 
=  \E\!\big[ \, e_{\last+M}^\transp e_{\last+M}\notrigvar{\last+M} \, \big| \, \Yall_\last, \Uall_\last \, \big] \nonumber \\
&\phantom{===}  - \E\!\big[\, e_{\last+M}^\transp e_{\last+M}\trigvar{\last+M} \, \big| \, \Yall_\last, \Uall_\last \, \big] \nonumber \\
&\phantom{=}= \norm{\eIhat_{\last+M|\last}}^2  - \norm{\eIIhat_{\last+M|\last}}^2 
+ \trace(\PI_{\last+M|\last} - \PII_{\last+M|\last}) \nonumber \\
&\phantom{=}= \trace( \PKF_{\last+M|\last} + \Xi_{\last,M} - \PKF_{\last+M} )
\end{align}
%
% Previous:
%\begin{align}
%&\trigsig_{\last+M |\last} 
%=  \E[ (\eI_{\last+M})^\transp \eI_{\last+M} | \Yall_\last, \Uall_\last ] - \E[ (\eII_{\last+M})^\transp \eII_{\last+M} | \Yall_\last, \Uall_\last  ] \nonumber \\
%&\phantom{=}= \norm{\eIhat_{\last+M|\last}}^2  - \norm{\eIIhat_{\last+M|\last}}^2 
%+ \trace(\PI_{\last+M|\last} - \PII_{\last+M|\last}) \nonumber \\
%&\phantom{=}= \trace( \PKF_{\last+M|\last} + \Xi_{\last,M} - \PKF_{\last+M} )
%\end{align}
where $\E[e^\transp e] 
%= (\E[e])^\transp \E[e] + \trace(\Var[e]) 
= \norm{\E[e]}^2 + \trace(\Var[e])$ with $\norm{\cdot}$ the Euclidean norm was used.
% E.g. Matrix-Cookbook 
% http://www.math.uwaterloo.ca/~hwolkowi/matrixcookbook.pdf
% Nov. 2012
% eq. (327)
\end{proof}

% LTI system (not required anymore, with new notation)
%
%For a time-invariant system \eqref{eq:sys_x} with $A_k = A$, $Q_k = Q$ for all $k$,
%%\begin{equation}
%%A_k = A, Q_k = Q, \quad \forall \, k,
%%\end{equation}
%the ST simplifies to 
%\begin{align}
%\gamma_{\last+M} = 1 \,\, \Leftrightarrow \,\, \trace\big( \Vo^M (\PKF_{\last}) - \PKF_{\last+M} \big) \geq C_{\last+M} 
%\label{eq:STsquaredErrorLTI}
%\end{align}
%where $\Vo^M$ denotes the $M$-times application of $\Vo(X) := A X A^\transp + Q$.

The self triggering rule is intuitive: a communication is triggered when the uncertainty of the open-loop estimator
%state predictions 
(prediction variance $\PKF_{\last+M|\last}+\Xi_{\last,M}$) 
exceeds the closed-loop uncertainty (KF variance $\PKF_{\last+M}$) by more than the cost of communication.  The estimation mean does not play a role here, since it is zero in both cases for $k = \lastel$.
%$\eI$ and $\eII$ have zero mean for $k = \lastel$.
% (\cf lemmas \ref{lem:PDF_eI} and \ref{lem:PDF_eII}).

%%%%%%%%%%%%%%%%%%%%%%%%%%%%%%%%%%%%%%%%%%%%%%%%%%%%%%%%%%%%%%%%%%%%%%%%%%%%%%%%
\subsection{Predictive trigger}
Similarly, we can employ lemmas \ref{lem:PDF_eII} and \ref{lem:PDF_eI} to compute the predictive trigger \eqref{eq:PTgeneral}.  
\begin{proposition}
For the quadratic error \eqref{eq:Ek_squares}, the predictive trigger (PT) \eqref{eq:PTgeneral} becomes, for $k > \lastel_{k-1}$, 
\begin{align}
&\gamma_{k+M} = 1 \,\, \Leftrightarrow \,\,  
\norm{ \bar{A}^M (\xKF_k  - \bar{A} \xP_{k-1} - B \xi_{k-1} ) }^2 \nonumber \\
&\qquad+ \trace\big( \PKF_{k+M|k} + \Xi_{k,M} - \PKF_{k+M} \big) \geq C_{k+M} \label{eq:PTsquaredError1}
\end{align}
and, for $k \leq \lastel_{k-1}$,
\begin{align}
&\gamma_{k+M} = 1 \,\, \Leftrightarrow \,\,  
\trace\big( \PKF_{\lastel+\Delta|\lastel} + \Xi_{\lastel,\Delta} - \PKF_{\lastel + \Delta} \big) \geq C_{\lastel + \Delta}  .
\label{eq:PTsquaredError2} 
\end{align}
with $\Delta$ as defined in Lemma~\ref{lem:PDF_eI}.
\end{proposition}
\begin{proof}
For $k > \lastel_{k-1}$ (\ie the last scheduled trigger occurred in the past), we obtain from lemmas \ref{lem:PDF_eII} and \ref{lem:PDF_eI}
\begin{align}
%\E[ & E_{k+M} | \Yall_k ] 
\trigsig_{k+M|k} 
%&=  \E[ (\eI_{k+M})^\transp \eI_{k+M} | \Yall_k, \Uall_\last ]  - \E[ (\eII_{k+M})^\transp \eII_{k+M} | \Yall_k, \Uall_\last ] \nonumber \\
&= \norm{ \bar{A}^M (\xKF_k - A \xP_{k-1} -B \xi_{k-1} ) }^2 \nonumber \\
&\phantom{=} 
+ \trace\big( \PKF_{k+M|k} + \Xi_{k,M} - \PKF_{k+M} \big),
\label{eq:PTsquaredError1_E}
\end{align}
where we used $\bar{A}^{k-\last} \xKF_{\last} + \sum_{m=1}^{k-\last} \bar{A}^{k-\ell-m} B \xi_{\ell+m-1}  = A \xP_{k-1} + B \xi_{k-1}$, which follows from the definition of the remote estimator \eqref{eq:remoteEst} with $\gamma_k = 0$ for $k > \last$.

Similarly, for $k \leq \lastel_{k-1}$,
% (\ie a trigger is scheduled now or in future), 
we obtain
$\trigsig_{k+M|k} = \trace\big(  \PKF_{\lastel + \Delta|\lastel} + \Xi_{\lastel,\Delta} - \PKF_{\lastel + \Delta} \big)$.
%\begin{align}
%%\E[ & E_{k+M} | \Yall_k ] \nonumber \\
%\trigsig_{k+M|k} &= \trace\big(  \PKF_{\lastel + \Delta|\lastel} - \PKF_{\lastel + \Delta} \big) .
%\label{eq:PTsquaredError2_E}
%\end{align}
\end{proof}

% LTI system (not required anymore, with new notation)
%
%For a time-invariant system, the triggering rules can again be simplified and are given, for $k > \lastel_{k-1}$, by
%\begin{align}
%&\gamma_{k+M} = 1 \,\, \Leftrightarrow \,\,  \nonumber \\
%&\norm{ A^M (\xKF_k \! -\! A \hat{x}_{k-1} ) }^2 
%%&\norm{ A^M (\xKF_k \! -\! A^{k-\last} \xKF_\last ) }^2 
% + \trace\big( \Vo^M (\PKF_k) \!-\! \PKF_{k+M} \big) \geq C_{k+M} \label{eq:PTsquaredErrorLTI1}
%% &\norm{ A^M (\xKF_k - A^K \xKF_\last ) }^2 
%% + \trace\big( \Vo^M (\PKF_k) - \PKF_{k+M} \big) \geq C_{k+M} \nonumber 
%\end{align}
%and, for $k \leq \lastel_{k-1}$, by
%\begin{align}
%&\gamma_{k+M} = 1 \,\, \Leftrightarrow \,\,  
%%\trace\big( \Vo^{k+M-\kappa} (\PKF_{\lastel}) - \PKF_{k+M} \big) \geq C_{k+M} .
%\trace\big( \Vo^{\Delta} (\PKF_{\lastel}) - \PKF_{\lastel + \Delta} \big) \geq C_{\lastel + \Delta} 
%%=  C_{k+M} 
%\label{eq:PTsquaredErrorLTI2} 
%\end{align}
%where $\Delta := k+M-\kappa$. 

%For a time-invariant system, this simplifies to, for $k > \lastel_{k-1}$,
%\begin{align}
%\E[ & E_{k+M} | \Yall_k ] \nonumber \\
%&= \norm{ A^M (\xKF_k - A^K \xKF_\last ) }^2 
% + \trace\big( \Vo^M (\PKF_k) - \PKF_{k+M} \big)
% \label{eq:PTsquaredErrorLTI1}
%\end{align}
%and, for $k \leq \lastel_{k-1}$,
%\begin{align}
%\E[ & E_{k+M} | \Yall_k ] 
%= \trace\big( \Vo^M (\PKF_{\lastel}) - \PKF_{k+M} \big) 
%\label{eq:PTsquaredErrorLTI2}
%\end{align}
%where $K:= k-\last$ in \eqref{eq:PTsquaredErrorLTI1} is the difference between last trigger and current time.

Similar to the ST \eqref{eq:STsquaredError}, the second term in the PT \eqref{eq:PTsquaredError1} relates the $M$-step open-loop prediction variance $\PKF_{k+M|k} + \Xi_{k,M}$ to the closed-loop variance $\PKF_{k+M}$.  However, now the reference time is the current time $k$, rather than the last transmission $\last$, because the PT exploits data until $k$.
%, while the ST makes the prediction without additional data after $\last$.  
In contrast to the ST, the PT also includes a mean term (first term in \eqref{eq:PTsquaredError1}).  When conditioning on new measurements $\Yall_k$ ($k>\last$), the remote estimator (which uses only data until $\last$) is biased; that is, the mean 
%of $\eI$ in 
%\eqref{eq:lem1_eIpdf}
\eqref{eq:lem1_eImean} 
is non-zero.  The bias term captures the difference in the mean estimates of the remote estimator ($A \xP_{k-1} + B \xi_{k-1}$) and the KF ($\xKF_k$), both predicted forward by $M$ steps.
This bias 
%is reflected in the first term in \eqref{eq:PTsquaredError1}
%, which captures the difference in the mean estimates, and 
contributes to the estimation cost \eqref{eq:PTsquaredError1}.

The rule \eqref{eq:PTsquaredError2} corresponds to the case where a trigger is already scheduled to happen at time  $\lastel$ in future (within the horizon $M$).  Hence, it is clear that the estimation error will be reset at $\lastel$, and from that point onward, variance predictions are used in analogy to the ST \eqref{eq:STsquaredError} ($\last$ replaced with $\lastel$, and the horizon $M$ with $\Delta$).  This trigger is independent of the data $\Yall_k$, $\Uall_k$ because the error at the future reset time $\lastel$ is fully determined by the distribution \eqref{eq:lem2_PDF_eII}, independent of $\Yall_k$, $\Uall_k$.

%%%%%%%%%%%%%%%%%%%%%%%%%%%%%%%%%%%%%%%%%%%%%%%%%%%%%%%%%%%%%%%%%%%%%%%%%%%%%%%%
\subsection{Discussion}
To obtain insight into the derived PT and ST, we next analyze and compare their structure.  To focus on the essential triggering behavior and simplify the discussion, we consider the case without inputs ($B_i \vik{u}{i}{k-1} = 0$ in \eqref{eq:sys_x}).
%In addition to the proposed triggers PT and ST, 
We also compare to an \emph{event trigger} (ET), 
%(\ie one with instantaneous triggering decision), 
which is obtained from the PT \eqref{eq:PTsquaredError1} by setting $M=0$:
%, which results in
\begin{align}
\gamma_{k} = 1 \,\, \Leftrightarrow \,\, 
\trigsig_{k|k}
&= \norm{ \xKF_k - A \xP_{k-1}  }^2 
%= \norm{ \xII_k - \xI_k  }^2
\geq C_k. \label{eq:ETsquaredError}
\end{align}
The trigger directly compares the two options at the remote estimator, $\xKF_k$ and $A \xP_{k-1}$.
%$\xI_k$ and $\xII_k$.
%the current local estimate to the remote estimate that results from not sending an update. 
To implement the ET, communication must be available instantaneously if needed.

The derived rules for ST, PT, and ET have the same threshold structure
% Cut (for space constraints)
%\footnote{For the ST \eqref{eq:STsquaredError}, \eqref{eq:trigGenStructure} is understood in the sense that \eqref{eq:trigGenStructure} is evaluated for increasing $M\geq 1$ until a positive trigger $\gamma_{k+M}=1$ is found.} 
\begin{equation}
\gamma_{k+M} = 1 \,\, \Leftrightarrow \,\,  \trigsig_{k+M|k} \geq C_{k+M} 
\label{eq:trigGenStructure}
\end{equation}
where the communication cost $C_{k+M}$ corresponds to the triggering threshold.
The triggers differ in the expected estimation cost $\trigsig_{k+M|k}$. 
% We next analyze the structure of the triggers in more detail.
% , where we concentrate on the time-invariant case for simplicity. \todo{?}
To shed light on this difference, 
%For the purpose of comparing the structure of the different triggers, 
we introduce
\begin{align}
\trigsigM_{k,M} &:= \norm{ A^M (\xKF_k \! -\! A \xP_{k-1} ) }^2 \label{eq:trigSigMean} \\
\trigsigV_{k,M} &:= \trace( \PKF_{k+M|k} \!-\! \PKF_{k+M} ). \label{eq:trigSigVar}
\end{align}
With this, the triggers 
%We then have the following characterization:  the triggers 
ST \eqref{eq:STsquaredError}, PT \eqref{eq:PTsquaredError1}, \eqref{eq:PTsquaredError2}, and ET \eqref{eq:ETsquaredError} 
%can then be characterized as follows.  Each trigger is 
are given by \eqref{eq:trigGenStructure} with
\begin{align}
\trigsig_{k+0|k} &= \trigsigM_{k,0}, M=0 && \text{(ET)} \label{eq:ETcharacterization} \\
\trigsig_{k+M|k} &= \trigsigM_{k,M} + \trigsigV_{k,M} \quad \quad \quad && \text{(PT), $k > \lastel$} \label{eq:PTcharacterization1} \\
\trigsig_{k+M|k} &= \trigsigV_{\kappa,\Delta} && \text{(PT), $k \leq \lastel$} \label{eq:PTcharacterization2} \\
\trigsig_{\last+M|\last} &= \trigsigV_{\last,M} && \text{(ST)} . \label{eq:STcharacterization}
\end{align}
Hence, the trigger signals are generally a combination of the `mean' signal \eqref{eq:trigSigMean} and the `variance' signal \eqref{eq:trigSigVar}.  Noting that the mean signal \eqref{eq:trigSigMean} depends on real-time measurement data $\Yall_k$ (through $\xKF_k$), while the variance signal \eqref{eq:trigSigVar} does not, we can characterize ET and PT as 
%\emph{closed-loop}
\emph{online triggers}, while ST is an 
%\emph{open-loop}
\emph{offline trigger}. This reflects the intended design of the different
triggers. ST is designed to predict the next trigger at the time $\last$ of the
last triggering, without seeing any data beyond $\last$.  This allows the sensor
to go to sleep in-between triggers, for example.  ET and PT, on the other hand, continuously monitor the sensor data to make more informed transmit decisions (as shall be seen in the following comparisons).
%when comparing the effectiveness of the different triggers in the following examples).

While ET requires instantaneous communication, 
%which restricts online allocation of communication resources,
which is limiting for online allocation of communication resources,
%which is impractical  when integration with the communication system is desired, 
PT makes the transmit decision $M\geq1$ steps ahead of time.  ET compares the mean estimates only (\cf \eqref{eq:ETcharacterization}), while PT results in a combination of mean and variance signal (\cf \eqref{eq:PTcharacterization1}).  If a transmission is already scheduled for $\lastel_{k-1} \geq k$, PT resorts to the ST mechanism for predicting beyond $\lastel_{k-1}$; that is, it relies on the variance signal only (\cf \eqref{eq:PTcharacterization2}).
  
While ST can be understood as an \emph{open-loop} trigger (\eqref{eq:STcharacterization} can be computed without any measurement data), ET clearly is a \emph{closed-loop} trigger requiring real-time data $\Yall_k$ for the decision on $\gamma_k$.  PT can be regarded as an intermediate scheme exploiting real-time data and variance-based predictions.  
Accordingly, the novel predictive triggering concept lies between the known concepts of event and self triggering.

The ST is similar to the variance-based triggers proposed in \cite{TrDAn14b}.  
Therein, it was shown for a slightly different scenario (transmission of measurements instead of estimates) that event triggering decisions based on the variance are independent of any measurement data and can hence be computed off-line.
%For a slightly different scenario with transmission of measurements (instead of local estimates), it was shown therein that event triggering decisions based on the variance of the measurement prediction error are independent of any measurement data and can be computed off-line.  
Similarly, when assuming that all problem parameters $A$, $H$, $Q$, $R$ in \eqref{eq:sys_x}, \eqref{eq:sys_y} are known a priori, 
%the variance signal \eqref{eq:trigSigVar} and thus 
\eqref{eq:STsquaredError} can be pre-computed for all times.  However, if some parameters only become available during operation (\eg the sensor accuracy $R_k$), the ST also becomes an online trigger. 
%In  \cite{TrDAn14b}, it was further shown that periodic transmit sequences result for variance-based triggering applied to a time-invariant process.
% \eqref{eq:sys_x}, \eqref{eq:sys_y}.
%  We shall also make this observation for the ST in the simulations in \sect \ref{sec:simulations}.

For the case with inputs ($B_i \vik{u}{i}{k-1} \neq 0$ in \eqref{eq:sys_x}), the triggering behavior is qualitatively similar.  The mean signal \eqref{eq:trigSigMean} will include the closed-loop dynamics $\bar{A}$ and the input $\xi_{k-1}$ corresponding to other agents, and the variance signal \eqref{eq:trigSigVar} will include the additional term $\Xi_{k,M}$ accounting for the additional uncertainty of not knowing the true input.
%Key characteristics of the proposed self-trigger (ST) and predictive trigger (PT) are illustrated through numerical simulations of stable and unstable scalar processes.

%% file: content/6_TriggerAnalysis.tex
\section{Illustrative Example}
% Analysis of triggers
\label{sec:analysisTriggers}
To illustrate the behavior of the obtained PT and ST, we present a numerical example.  We study simulations of the stable, scalar, linear time-invariant (LTI) system \eqref{eq:sys_x}, \eqref{eq:sys_y} with:
% without inputs:
\begin{example}
%System \eqref{eq:sys_x}, \eqref{eq:sys_y}, with 
$A = 0.98$, $B=0$ (no inputs), $H = 1$, $Q = 0.1$, $R = 0.1$, and $\bar{x}_0 = X_0 = 1$.
\label{ex:ex1}
\end{example}

%%%%%%%%%%%%%%%%%%%%%%%%%%%%%%%%%%%%%%%%%%%%%%%%%%%%%%%%%%%%%%%%%%%%%%%%%%%%%%%%
\subsection{Self trigger}
We first consider the self trigger (ST).
Results of the numerical simulation of the event-based estimation system (\cf \fig \ref{fig:remoteEstimation}) consisting of the local state estimator \eqref{eq:KF1}--\eqref{eq:KF5}, the remote state estimator \eqref{eq:remoteEst}, and the ST \eqref{eq:STsquaredError} with constant cost $C_k = C= 0.6$ are shown in \fig \ref{fig:example1_1}.
\begin{figure}[tb]
\centering
\includegraphics[scale=0.95]{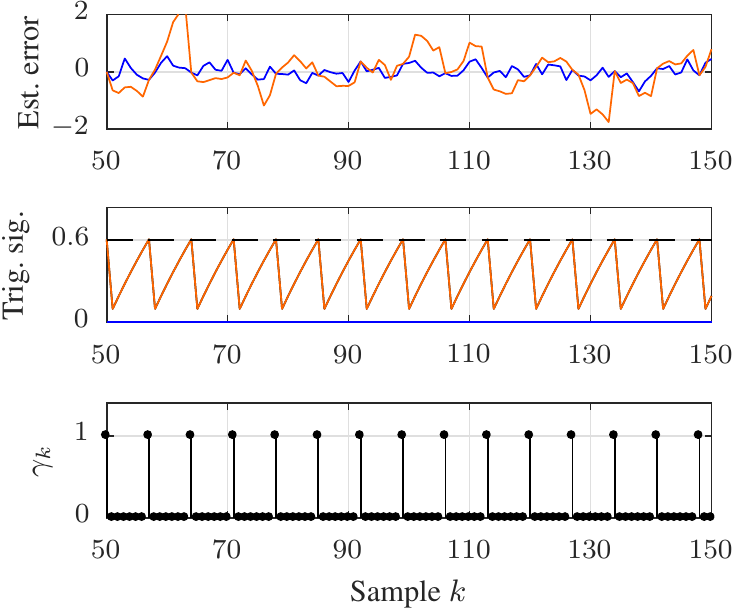}	% was 0.88, 0.9
\caption{\Ex~\ref{ex:ex1} with self trigger (ST).  
%The TOP graph shows the true state $x$ (\graph{black}), the KF estimate $\xKF$ (\graph{blue}), and the remote estimate $\hat{x}$ (\graph{orange}); and in the SECOND graph are the corresponding errors $\eKF=x-\xKF$ (\graph{blue}) and $e=x-\hat{x}$ (\graph{orange}).  
TOP: KF estimation error $\eKF=x-\xKF$ (\graph{blue}) and remote error $e=x-\hat{x}$ (\graph{orange}).
MIDDLE: 
%The THIRD graph shows 
components of the triggering signal 
$\trigsigM$ \eqref{eq:trigSigMean} (\graph{blue}), $\trigsigV$ \eqref{eq:trigSigVar} (\graph{black}, hidden), the triggering signal $\trigsig = \trigsigM + \trigsigV$ (\graph{orange}), and the threshold $C_k = 0.6$ (\graph{dashed}).  
BOTTOM: triggering decisions $\gamma$.
%The BOTTOM graph indicates the triggering decisions $\gamma$.
}
\label{fig:example1_1}
\end{figure}
The estimation errors of the local and remote estimator are
compared in the first graph.
As expected, the remote estimation error $e_k = x_k-\hat{x}_k$ (orange) is larger than the local estimation error $\eKF_k = x_k-\xKF_k$ (blue).  Yet, the remote estimator only needs 14\% of the samples.  

The triggering behavior is illustrated in the second graph showing the
triggering signals $\trigsigM$ \eqref{eq:trigSigMean}, $\trigsigV$ \eqref{eq:trigSigVar}, and $\trigsig = \trigsigM + \trigsigV$, and the bottom graph depicting the triggering decision $\gamma$.
%As can be seen from the third graph in \fig \ref{fig:example1_1}, 
Obviously, the ST entirely depends on the variance signal $\trigsigV$ (orange, identical with $\trigsig$ in black), while $\trigsigM = 0$ (blue).  This reflects the previous discussion about the ST being independent of online measurement data.  
The triggering behavior (the signal $\trigsig$ and the decisions $\gamma$) is actually  \emph{periodic}, which can be deduced as follows:
%Furthermore, it can be seen that the triggering signal and thus also the transmit decisions $\gamma_k$ are \emph{periodic}.  This can be deduced as follows: 
%From KF theory \cite{AnMo05}, 
the variance $\PKF_k$ of the KF \eqref{eq:KF1}--\eqref{eq:KF5} converges exponentially to a steady-state solution $\PKFss$, \cite{AnMo05}; hence, the triggering law \eqref{eq:STsquaredError} asymptotically becomes  
$\trace( \Vo^M(\PKFss) - \PKFss ) \geq C$ with $\Vo(X) := AXA^\transp + Q$,
% Cut (space):
%\begin{align}
%\trace( \Vo^M(\PKFss) - \PKFss ) \geq C, \quad  \Vo(X) := AXA^\transp + Q
%\label{eq:STsquaredErrorLTI_ex1}
%\end{align}
and \eqref{eq:STsquaredError} thus has a unique 
%(time-invariant) 
solution $M$ corresponding to the period seen in \fig \ref{fig:example1_1}.  
%Moreover, the triggering decision becomes deterministic in the that it is independent of any real-time data and can be computed off-line.

Periodic transmit sequences are typical for variance-based triggering on time-invariant problems, which has also been found and formally proven for related scenarios in \cite{TrDAn14b,LeDeQu15}.  
%Obviously, there is no such periodic solution for a general time-varying process \eqref{eq:sys_x}, \eqref{eq:sys_y}.

%%%%%%%%%%%%%%%%%%%%%%%%%%%%%%%%%%%%%%%%%%%%%%%%%%%%%%%%%%%%%%%%%%%%%%%%%%%%%%%%
\subsection{Predictive trigger}
The results of simulating \Ex \ref{ex:ex1}, now with 
%For the same example, \Ex \ref{ex:ex1}, the simulation results of the event-based estimation system with 
the PT \eqref{eq:PTsquaredError1}, \eqref{eq:PTsquaredError2}, and prediction horizon $M=2$, are presented in \fig \ref{fig:example1_2} for the cost $C_k = C = 0.6$, and in \fig \ref{fig:example1_3} for $C_k = C = 0.25$.  Albeit using the same trigger, the two simulations show fundamentally different triggering behavior: while the triggering signal $\trigsig$ and the decisions $\gamma$ in \fig \ref{fig:example1_2} are irregular, they are periodic in \fig \ref{fig:example1_3}.  
\begin{figure}[tb]
\centering
\includegraphics[scale=0.95]{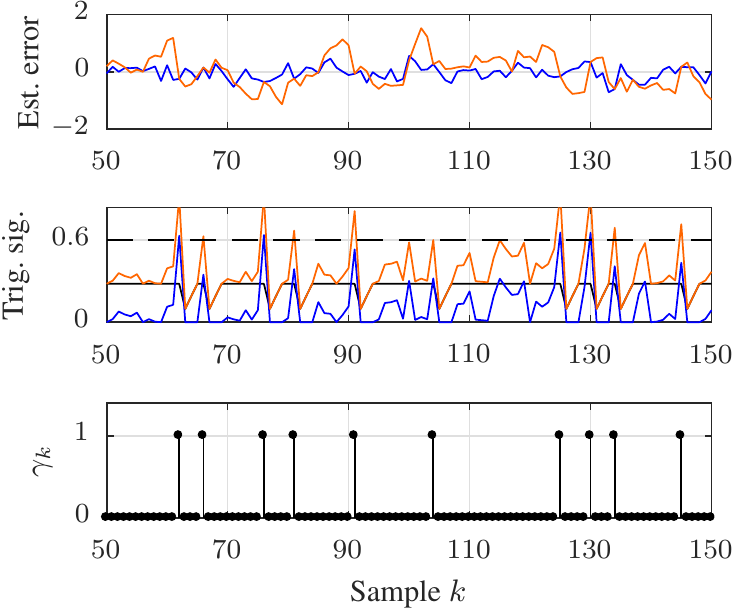}
\caption{\Ex~\ref{ex:ex1} with predictive trigger (PT) and $C_k = 0.6$.  Coloring of the signals is the same as in \fig \ref{fig:example1_1}.  The triggering behavior is \emph{stochastic}.
}
\label{fig:example1_2}
\end{figure}
\begin{figure}[tb]
\centering
\includegraphics[scale=0.95]{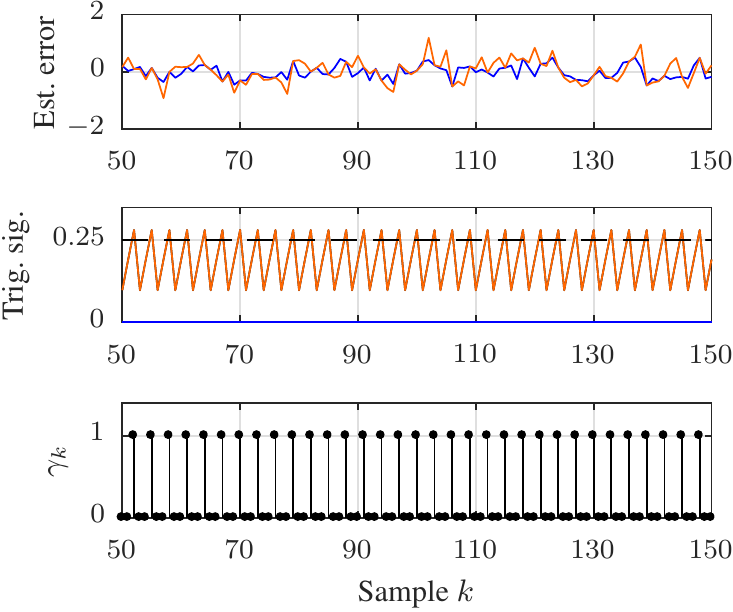}
\caption{\Ex~\ref{ex:ex1} with predictive trigger (PT) and $C_k = 0.25$.  Coloring of the signals is the same as in \fig \ref{fig:example1_1}. The triggering behavior is \emph{periodic}.
}
\label{fig:example1_3}
\end{figure}

Apparently, the choice of the cost $C_k$ determines the different behavior of the PT.
For $C_k = 0.6$, the triggering decision depends on both, the mean signal $\trigsigM$ and the variance signal $\trigsigV$, as can be seen from \fig \ref{fig:example1_2} (middle graph).
%From \fig \ref{fig:example1_2} (third graph), it can be seen that for $C_k = C = 0.6$, the triggering decision depends on both, the mean signal $\chi_\text{mean}$ and the variance signal $\chi_\text{var}$.  
Because $\trigsigM$ is based on real-time measurements,
% (through $\xKF_k$ in \eqref{eq:PTsquaredError1}), 
which are themselves random variables \eqref{eq:sys_y}, the triggering decision is a random variable.  
We also observe in \fig \ref{fig:example1_2} that the variance signal $\trigsigV$ is alone not sufficient to trigger  a communication.
% Cut (space): \footnote{After convergence of the local estimator variance $\PKF_k$, $\trigsigV$ corresponds to \eqref{eq:STsquaredErrorLTI_ex1}, which does not exceed the chosen $C_k$ for $M=2$ iterations.}
However, when lowering the cost of communication $C_k$ enough, the variance signal alone becomes sufficient to cause triggers.  Essentially, triggering then happens 
%the prediction 
according to \eqref{eq:PTsquaredError2} only, and  \eqref{eq:PTsquaredError1} becomes irrelevant.  Hence, the PT resorts to self triggering behavior for small enough communication cost $C_k$.  That is, the PT undergoes a phase transition for some value of $C_k$ from stochastic/online triggering to deterministic/offline triggering behavior.

%While, for a stationary LTI system with constant cost $C_k=C$, the self-trigger is deterministic\footnote{Deterministic in the sense that \eqref{eq:STsquaredErrorLTI_ex1}, which is independent of real-time data and can be computed off-line, fully determines the periodic triggering sequence.}  for all $C$,
%the predictive trigger undergoes a phase transition for some value of $C$ from stochastic/on-line triggering decision to deterministic/off-line behavior.

\subsection{Estimation versus communication trade-off}
Following the approach from \cite{TrCa15}, we evaluate the effectiveness of different triggers by comparing their trade-off curves of average estimation error $\Ec$ versus average communication $\Cc$ obtained from Monte Carlo simulations.  In addition to the ST \eqref{eq:STsquaredError} and the PT \eqref{eq:PTsquaredError1}, \eqref{eq:PTsquaredError2}, $M=2$, we also compare against the ET \eqref{eq:ETsquaredError}.  The latter is expected to yield the best trade-off because it makes the triggering decision at the latest possible time (ET decides at time $k$ about communication at time $k$).  
%Yet, ET assumes that communication is available instantaneously if needed, and otherwise wasted, which is often impractical as argued in the introduction.

%The estimation error $\Ec$ versus average communication $\Cc$ curves are obtained from Monte Carlo simulations.
The estimation error $\Ec$ is measured as the squared error $e_k^2$ averaged
over the simulation horizon (200 samples) and \num{50000} simulation runs. 
The average communication $\Cc$ is normalized such that $\Cc=1$ means
$\gamma_k=1$ for all $k$, and $\Cc=0$ means no communication (except for one
enforced trigger at $k=1$).
By varying the constant communication cost $C_k = C$ in a suitable range, an $\Ec$-vs-$\Cc$ curve is obtained, which represents the estimation/communication trade-off for a particular trigger.  The results 
%from simulations of 
for \Ex \ref{ex:ex1} are shown in \fig \ref{fig:example1_EvsC}.
\begin{figure}[tb]
\centering
\includegraphics[scale=0.95]{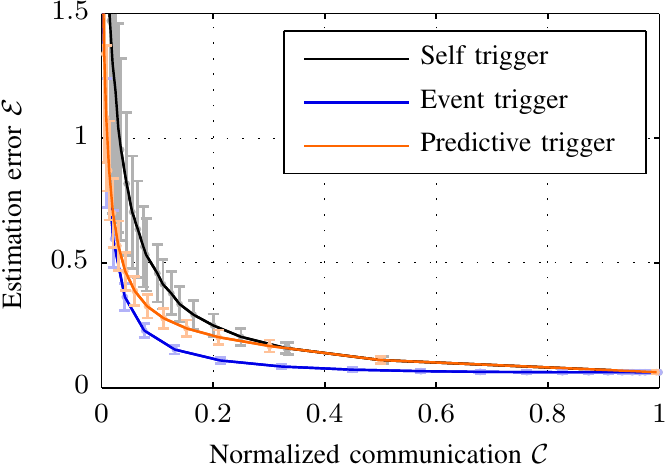} % was: .9
\caption{Trade-off between estimation error $\Ec$ and average communication $\Cc$ for different triggering concepts applied to \Ex \ref{ex:ex1}. 
Each point represents the average from 50'000 Monte Carlo simulations, and the light error bars correspond to one standard deviation. 
%It can be seen that the novel concept of predictive triggering provides a middle ground between event triggering and self triggering. 
}
\label{fig:example1_EvsC}
\end{figure}

%Figure \ref{fig:example1_EvsC} compares the $\Ec$-vs-$\Cc$ curves for three different triggering concepts.
%In addition to the self trigger (ST) \eqref{eq:PTgeneral}, \eqref{eq:STsquaredErrorLTI} and the predictive trigger (PT) \eqref{eq:PTgeneral}, \eqref{eq:PTsquaredErrorLTI1}, \eqref{eq:PTsquaredErrorLTI2} with $M=2$, we also included simulations with an event trigger (ET).  The latter is obtained from \eqref{eq:PTsquaredErrorLTI1}, \eqref{eq:PTsquaredErrorLTI2} by setting $M=0$, which results in
%\begin{align}
%\E[  E_{k} | \Yall_k ] = \norm{ \xKF_k - A^K \xKF_\last  }^2 = \norm{ \xKF_k - A \hat{x}_{k-1}  }^2 .
% \label{eq:ETsquaredErrorLTI}
%\end{align}
%That is, the trigger compares the \emph{current} local to the remote estimate (in case no update is sent); it assumes that communication is available \emph{instantaneously} if needed.

Comparing the three different triggering schemes, we see that the ET is superior, as expected, because its curve is uniformly below the others.  Also expected, the ST is the least effective 
%in making the trade-off between estimation performance and communication requirements 
since no real-time information is available and triggers are purely based on variance predictions.  
%This makes perfect sense: the event-trigger can exploit the maximum available information for making a transmit decision (actually all data until and including the decision time), while the self-trigger decision is purely based on predictions.
% and thus incurs more uncertainty in the decision making.  
The novel concept of predictive triggering can be understood as an intermediate solution between these two extremes.  For small communication cost $C_k$ (and thus relatively large communication $\Cc$), the PT behaves like the ST, as was discussed in the previous section and is confirmed in \fig \ref{fig:example1_EvsC} (orange and black curves essentially identical for large $\Cc$).
%as discussed in the previous section and as can be seen in \fig \ref{fig:example1_EvsC} (orange and black curves identical for large $\Cc$).  
When the triggering threshold $C_k$ is relaxed (\ie the cost increased), the PT also exploits real-time data for the triggering decision (through \eqref{eq:trigSigMean}), similar to the ET.  Yet, the PT must predict the decision $M$ steps in advance making its $\Ec$-vs-$\Cc$ trade-off generally less effective than the ET.  In \fig \ref{fig:example1_EvsC}, the curve for PT is thus between ET and ST and approaches either one of them for small and large communication $\Cc$.
%; for small $\Cc$, it approaches the performance of ET, and for large $\Cc$, the performance of ST.  

% Cut (space):
%The same qualitative behavior of the different triggering mechanisms is observed for an unstable process; results can be found in \cite{Tr16}.

%Fig.\ \ref{fig:example2_EvsC} shows the $\Ec$-vs-$\Cc$ curves for the unstable system:
%\begin{example}
%%System \eqref{eq:sys_x}, \eqref{eq:sys_y}, with 
%$A_k = 1.1$, $H_k = 1$, $Q_k = 0.1$, $R_k = 0.1$ for all $k$, and $\bar{x}_0 = X_0 = 1$.
%\label{ex:ex2}
%\end{example}
%The same qualitative behavior of the different triggering mechanisms as in \fig \ref{fig:example1_EvsC} can be observed.
%\begin{figure}[tb]
%\centering
%\includegraphics[scale=0.98]{figure_example1_EvsC} % was: .9
%\caption{Estimation-vs-communication trade-offs for investigated triggering concepts applied to \Ex \ref{ex:ex2} (unstable process).
%}
%\label{fig:example2_EvsC}
%\end{figure}

%% file: content/7_Experiments.tex
\section{Hardware Experiments: Remote Estimation \& Feedback Control}
\label{sec:experiments}
Experimental results of applying the proposed PT and ST on an inverted pendulum platform are presented in this section.  We show that trade-off curves in practice are similar to those in simulation (\cf \fig \ref{fig:example1_EvsC}), and that the triggers are suitable for feedback control (\ie stabilizing the pendulum).

%The proposed framework has been validated in real hardware experiments. In this
%section, the results of this experiments, including pure state estimation and
%feedback control based on PT and ST, will be presented.

\subsection{Experimental setup}
The experimental platform used for the experiments of this section is the
inverted pendulum system shown in Fig.~\ref{fig:cart-pole}. Through appropriate
horizontal motion, the cart can stabilize the pendulum in its upright position
($\theta = \SI{0}{\radian}$).
The system state is given by the position and velocity of the cart, and angle
and angular velocity of the pole, \ie $x = (s, \dot{s}, \theta,
\dot{\theta})^\transp$.
The cart-pole system is widely used as a
benchmark in control~\cite{Bou13} because it has nonlinear, fast, and unstable
dynamics.

%The experimental platform used for the validation was an inverted pendulum on a
%cart.
%Such a cart-pole system is shown in Fig.~\ref{fig:cart-pole}. The cart can move horizontally and like that stabilize
%the pendulum in an upright position. The cart-pole system is widely used as a
%benchmark in control theory~\cite{Bou13} as it is nonlinear and has fast
%dynamics.

\begin{figure}
\centering
\includegraphics[scale=0.05]{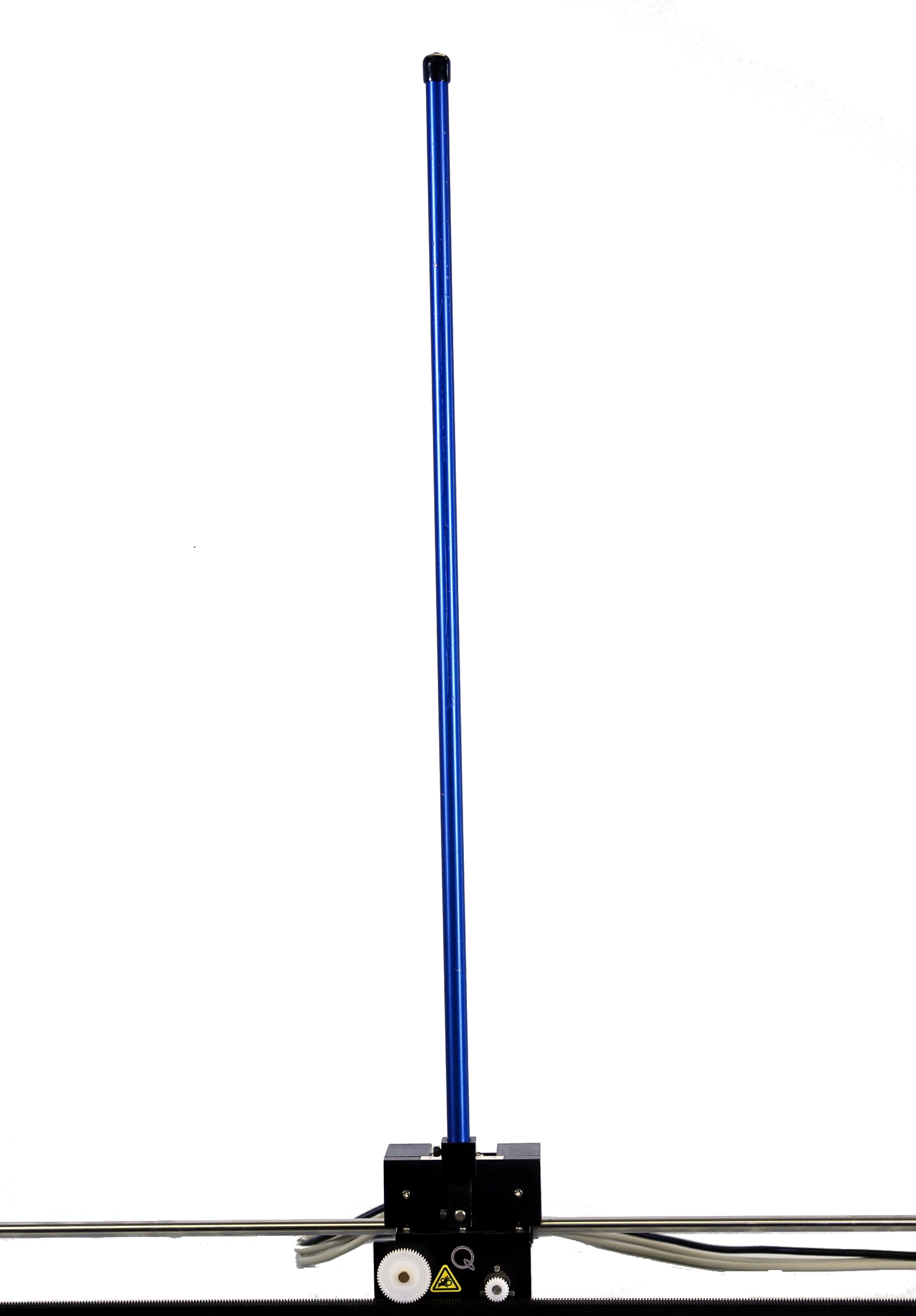}
\hspace{0.5cm}
\includegraphics{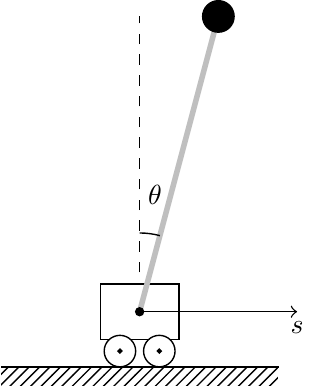}
\caption{Picture and schematic of the cart-pole system used for the
experiments.}
\label{fig:cart-pole}
\end{figure}

The sensors and actuator of the pendulum hardware are connected through data acquisition devices to a standard laptop running Matlab/Simulink.  Two encoders measure the angle $\theta_k$ and cart position $s_k$ every \SI{1}{\milli\second}; and voltage $u_k$ is commanded to the motor with the same update interval.
The full state $x_k$ can be constructed from the sensor measurements through finite differences.
The triggers, estimators, and controllers
%, as well as network models (\eg packet drops) 
are implemented in Simulink.  The pendulum system thus represents one `Thing $i$' of \fig \ref{fig:iotControlSchematic}.  
%Communication with local sensors and actuator is periodically at \SI{1}{\milli\second}.

%For the trigger concepts presented here the goal is to estimate the state of
%the system, which consists of the position $s$ of the cart and its velocity
%$\dot{s}$, as well as the angle $\theta$ of the pendulum and the angular
%velocity $\dot{\theta}$. The position and the angle can be measured with encoder sensors.

%The desired state, when the pendulum is in an upright position, is denoted
%with $\theta = \SI{0}{\degree}$. This is an unstable equilibrium of the system
%and therefore a controller needs to be used. For the first experiments the
%pendulum will be stabilized with a local controller, which directly receives the
%measurements of the encoders and the estimates of the derivatives from a finite
%difference observer with a sample time of \SI{1}{\milli\second}. The state of
%the stabilized system is then estimated on a remote estimator (\cf
%Fig.~\ref{fig:remoteEstimation}) using ET, ST and PT.

As the upright equilibrium is unstable, a stabilizing feedback controller is
needed.  We employ a linear-quadratic regulator (LQR), which is a standard design for multivariate feedback control,
\cite{anderson2007optimal}.  Assuming linear dynamics (with a model as given in \cite{Quanser}) and perfect state measurements, a linear state-feedback controller, $u_k = F x_k$, is obtained as the \emph{optimal} feedback controller that minimizes a quadratic cost function
\begin{align}
J = \lim_{K \to \infty} \frac{1}{K} \E\!\Big[ \sum\limits\nolimits_{k=0}^{K-1} x_k^{\mathrm{T}}Qx_k + u_k^\mathrm{T}Ru_k \Big].
\end{align}
The positive definite matrices $Q$ and $R$ are design parameters, which represent the designer's trade-off in achieving a fast response (large $Q$) or low control energy (large $R$). 
Here, we chose $Q = 30 I$ and $R = I$ with $I$ the identity matrix, which leads to stable balancing with slight motion of the cart.
Despite the true system being nonlinear and state measurements not perfect, LQR leads to good balancing performance, 
%LQR is a typical controller for balancing problems \cite{TrDAn12b,MaRiSa14}, 
which has also been shown in previous work on this platform \cite{marco_ICRA_2017_etal}.

%The cart-pole is a nonlinear system, which may however for small
%oscillations of the pendulum be linearized. The parameters of the
%system matrices for the experimental setup used here were derived
%in~\cite{Quanser}.

%The controller was obtained from a lest-squares approach. A linear-quadratic
%regulator (LQR) is derived by minimizing the cost
%function~\cite{anderson2007optimal} 
%\begin{align}
%J = \sum\limits_0^\infty \left(x_k^{\mathrm{T}}Qx_k + u_k^\mathrm{T}Ru_k\right).
%\end{align}
%This class of controllers was successfully used for pole balancing
%in~\cite{marco_ICRA_2017}. The matrices $Q$ and $R$ are design parameters. The
%diagonal elements of $Q$ indicate, how fast the corresponding state variable converges to \num{0}, high
%values in the $R$ matrix lead to low values of the input and a slower control.
%Here the $Q$ matrix was chosen as an identity matrix multiplied by \num{30} and $R$ was chosen to \num{1}.
%This leads to a stable system, but the cart is constantly moving. Other controllers
%leading to less movement of the cart would have simplified the task of state
%estimation what was not intended here.

Characteristics of the communication network to be investigated are implemented in the Simulink model.
The round time of the network 
%connecting the remote estimator with the local system 
is assumed to be \SI{10}{\milli\second}. 
For the PT, the prediction horizon is $M \!=\! 2$. 
Thus, the communication network has \SI{20}{\milli\second} to reconfigure,
which is expected to be sufficient 
%enough time 
for fast protocols such as~\cite{FerZimLWB}.

%Both the local and the remote
%estimator use this as a sample time to make predictions of the system state. The local Kalman
%filter which runs at the plant and whose output is sent to the remote estimator
%in case communication is needed, does not depend on the communication network.
%Therefore this Kalman filter is designed with a sample time of
%\SI{1}{\milli\second} to be able to provide the best estimates of the current
%state possible whenever there is a need for communicating it.
 
%For the PT the prediction horizon is, as before, two discrete
%time steps.
%Like that, the communication network has \SI{20}{\milli\second} to reconfigure,
%which is enough time for fast network protocols such as~\cite{FerZimLWB}.
 
\subsection{Remote estimation}
\label{sec:experimentsRemoteEstimation}
The first set of experiments investigates the remote estimation scenario as in \fig \ref{fig:remoteEstimation}.  For this purpose, the pendulum is stabilized locally via the above LQR, which runs at \SI{1}{\milli\second} and directly acts on the encoder 
measurements and their derivatives obtained from finite differences.  The closed-loop system thus serves as the dynamic process in \fig \ref{fig:remoteEstimation} (described by equation \eqref{eq:sys_x}),
% represented by \eqref{eq:sys_x}, 
whose state is to be estimated and communicated via ET, PT, and ST to a remote location, which could represent another agent from \fig \ref{fig:iotControlSchematic}.  

The local \emph{State Estimator} in \fig \ref{fig:remoteEstimation} is
implemented as the KF \eqref{eq:KF1}--\eqref{eq:KF5} with properly tuned
matrices and updated every \SI{1}{\milli\second} (at every sensor update).  
Triggering decisions are made at the round time of the network (\SI{10}{\milli\second}).
Accordingly, state predictions \eqref{eq:remoteEst} are made every \SI{10}{\milli\second} (in \emph{Prediction Thing $i$} in \fig \ref{fig:remoteEstimation}).

Analogously to the numerical examples in Sec.~\ref{sec:analysisTriggers}, we
investigate the estimation-versus-communication trade-off achieved by ET, PT, and ST. As can be
seen in Fig.~\ref{fig:errVScommCartPole}, all three triggers lead to
approximately the same curves. 
% and thus enabling the communication
%network to reconfigure or shut down resources, does not lead to a mentionable
%disadvantage.  
These results are qualitatively different from those of the numerical example in
\fig \ref{fig:example1_EvsC}, which showed notable differences between the
triggers.  Presumably, the reason for this lies in the low-noise environment of
this experiment.  The main source of disturbances is the encoder quantization,
which is negligible.  Therefore, the system is almost deterministic, and
predictions are very accurate.
Hence, in this setting, predicting future communication needs (PT, ST) does not involve any significant disadvantage compared to instantaneous decisions (ET).

\begin{figure}
\centering
\includegraphics{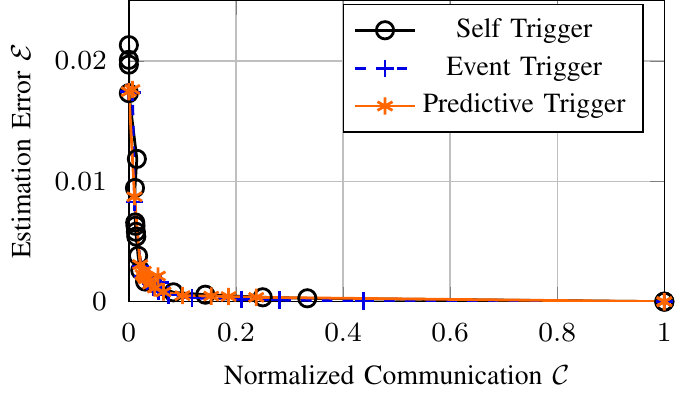}
\caption{Trade-off between averaged communication and the estimation error for a pendulum experiment with
low sensor noise.
Each marker represents the mean of \num{10} experiments with the same communication cost. The variance is negligible and thus omitted.}
\label{fig:errVScommCartPole}
\end{figure} 

%Each marker in Fig.~\ref{fig:errVScommCartPole} shows the mean of \num{10}
%experiments with the same communication cost. The variance between the different
%experiments was negligible small and is therefore omitted.

To confirm these results, we added zero-mean Gaussian noise with variance \num{5e-6} to the position and angle measurements.  This emulates analog angle sensors instead of digital encoders and is representative for many sensors in practice that involve stochastic noise. The results of this experiment are shown in
Fig.~\ref{fig:errVScommCartPole_addNoise}, which shows the same qualitative difference between the triggers as was observed in the numerical example in \fig \ref{fig:example1_EvsC}.

%The result of all triggers performing approximately the same way contradicts
%the results of the numerical example. The reason for this is seen in the
%accuracy of the encoder measurements, and the greatest possible
%extent of disturbances in the laboratory environment.  The results of the
%numerical experiment can be confirmed by adding some artificial Gaussian noise
%with mean \num{0} and variance \num{5e-6} to the position and angle measurements
%the Kalman filter receives. This might be the case, if instead of digital analog
%encoders were used. The results of such an experiment can be seen in
%Fig.~\ref{fig:errVScommCartPole_addNoise}. As for the numerical example, the PT
%is located between the ET- and the ST.

\begin{figure}
\centering 
\includegraphics{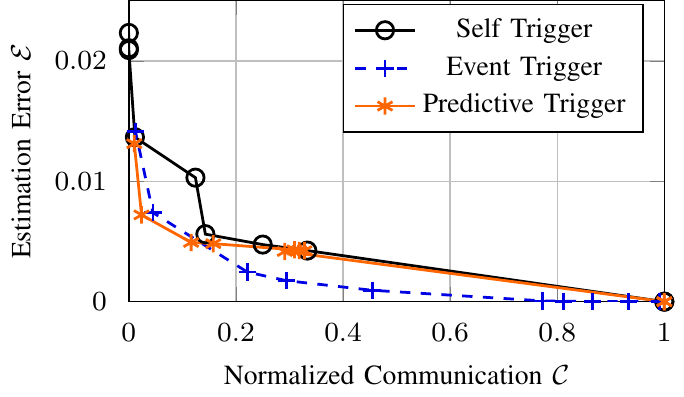}
\caption{Same experiment as in \fig \ref{fig:errVScommCartPole}, but with noisy sensors.
%Tradeoff between the normalized cost and the estimation error with
%noisy sensor inputs.
}
\label{fig:errVScommCartPole_addNoise}
\end{figure} 

\subsection{Feedback control}
\label{sec:pend_feedback}
The estimation errors obtained in Fig.~\ref{fig:errVScommCartPole} are fairly
small even with low communication. Thus, we expect the estimates obtained with PT and ST also to be suitable for feedback control, which we investigate here.  
In contrast to the setting in \sect \ref{sec:experimentsRemoteEstimation}, the LQR controller does not use the local state measurements at the fast update interval of \SI{1}{\milli\second}, but the state predictions \eqref{eq:remoteEst} instead.  This corresponds to the controller being implemented on a remote agent, which is relevant for IoT control as in \fig \ref{fig:iotControlSchematic}, where feedback loops are closed over a resource-limited network.

Figures \ref{fig:feedback_pred} and \ref{fig:feedback_self} show experimental results of using PT and ST for feedback control.
For these experiments, the weights of the LQR approach were chosen as those suggested by the manufacturer in~\cite{Quanser}, which leads to a slightly  more robust controller.
Both triggers are able to stabilize the pendulum well and save around \SI{80}{\percent} of communication.  

%The estimation errors obtained in Fig.~\ref{fig:errVScommCartPole}  are fairly
%small. Thus the estimate of the PT and the ST can not only be used
%for synchronization tasks, but even for stabilizing the pendulum.
%Results of such an experiment can be seen in Fig.~\ref{fig:feedback_pred}
%and~\ref{fig:feedback_self} for the PT and
%the ST respectively. Both triggers are able to stabilize the
%pendulum and save around \SI{80}{\percent} of communication. 

%For this setting a more robust controller was needed. Therefore, the weight
%matrices of the LQR approach were chosen as the ones proposed
%in~\cite{Quanser}.

In addition to disturbances inherent in the system, the experiments also include impulsive disturbances on the input (impulse of \SI{2}{\volt} amplitude and \SI{500}{\milli\second} duration every \SI{10}{\second}), which we added to study the triggers' behavior under deterministic disturbances. 
In addition to stochastic noise, such disturbances are relevant in many practical IoT scenarios (\eg a car braking, a wind gust on a drone).
%If PT or ST is used for feedback control in
%non-laboratory setups, it has to be robust not only in case of sensor noise
%but also in case of external, deterministic disturbances. For the experiments
%shown in Fig.~\ref{fig:feedback_pred} and~\ref{fig:feedback_self} such a
%deterministic disturbance was introduced.
%Every \SI{10}{\second} a disturbance of \SI{2}{\volt} amplitude and
%\SI{500}{\milli\second} length was added to the input signal of the cart. 
Under these disturbances, a particular advantage of PT over ST becomes apparent. The ST is an offline trigger, which yields periodic communication (in this setting) and does not react to the external disturbances.
The PT, on the other hand, takes
%PT not only takes the variance, but also 
the current error into account and is thus able to issue additional communication in case of disturbances. 
%Therefore, the average communication increases when the
%disturbance is detected. 
As a result, the maximum angle of the pendulum stays around \SI{0.03}{\radian}
in magnitude for the PT, while it comes close to \SI{0.04}{\radian} for the ST.

\begin{figure}
\centering
\includegraphics{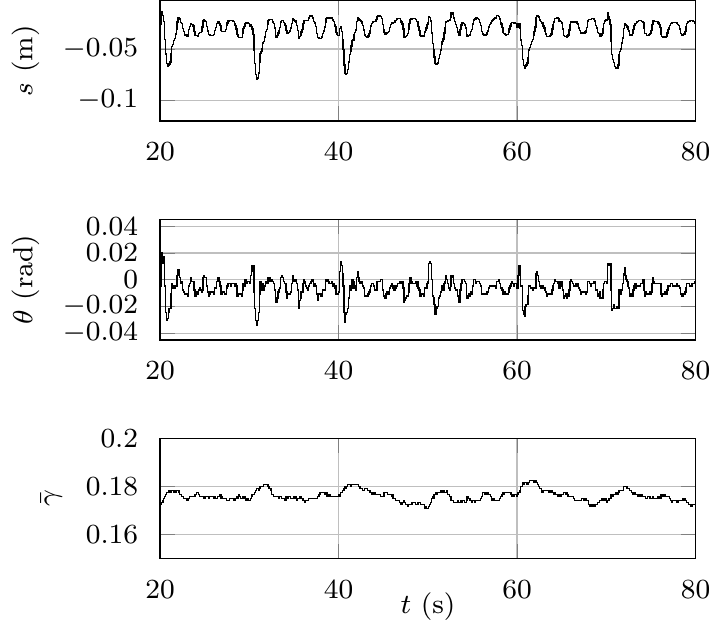}
\caption{Closing the feedback loop with the PT. The graphs show, from top to bottom, the cart
position $s$, the pendulum angle $\theta$, and the obtained average
communication $\bar{\gamma}$, computed as a moving average over \num{1200}
samples. The communication cost was set to $C_k = C = \num{0.009}$.
}
\label{fig:feedback_pred}
\end{figure}

\begin{figure}
\centering
\includegraphics{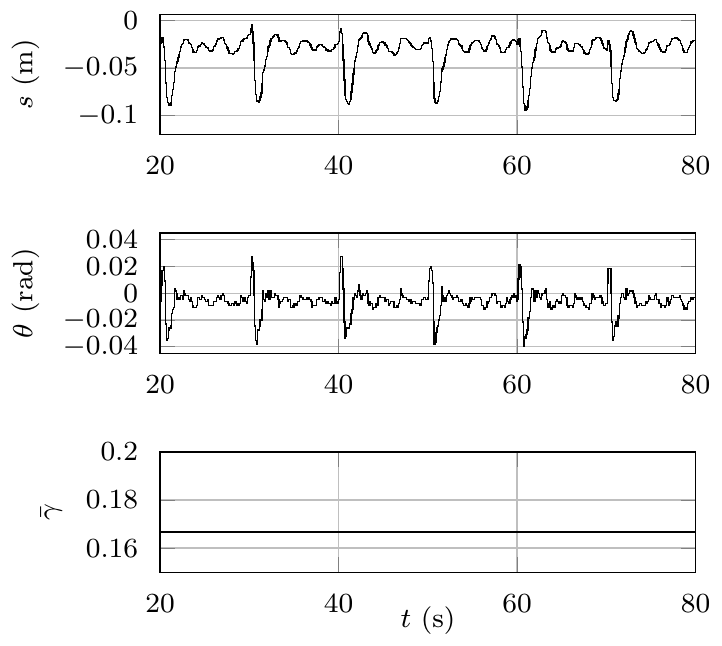}
\caption{Closing the feedback loop with the ST. Same plots as in \fig \ref{fig:feedback_pred}.
%Plotted are the
%position $s$ and the angle $\theta$ of the cart-pole system. The average
%communication $\bar{\gamma}$, computed with a moving average, is shown in the
%bottom plot.
}
\label{fig:feedback_self}
\end{figure}

%% file: content/MultiAgentScenario.tex
\section{IoT Control with Multiple Agents}
\label{sec:multiAgent}
In the preceding sections, we addressed the problem posed in \sect \ref{sec:objective}
%of predictive and self triggering 
for the case of two agents.  In this section, we discuss how these results can be used for the IoT scenario with multiple agents in \fig \ref{fig:iotControlSchematic}.  
Moreover, we sketch how the resulting closed-loop dynamics can be analyzed when remote estimates are used for feedback control.

Because we discuss multiple agents, we reintroduce the index `$i$' to refer to an individual agent $i$ from here onward.

\subsection{Multiple agents}
The developments for a pair of agents as in \fig \ref{fig:agentSchematic} in the previous sections equally apply to the IoT scenario in \fig \ref{fig:iotControlSchematic}. 
Each agent implements the blocks from \fig \ref{fig:agentSchematic}: \emph{State Estimation} is given by the KF \eqref{eq:KF1}--\eqref{eq:KF5}, \emph{Prediction} by \eqref{eq:remoteEst}, \emph{Control} by \eqref{eq:controlLawPre}, and the \emph{Event Trigger} is replaced by either the ST \eqref{eq:STsquaredError} or the PT \eqref{eq:PTsquaredError1}, \eqref{eq:PTsquaredError2}.  
In particular, each agent makes predictions for those other agents whose state it requires for coordination.  
Whenever one agent transmits its local state estimate, it is broadcast over the network and received by all agents that care about this information, \eg via many-to-all communication.
In the considered scenario, the dynamics of the things are decoupled according to \eqref{eq:sys_x}, \eqref{eq:sys_y} (\cf \fig \ref{fig:iotControlSchematic}), but their action is coupled through the cooperative control \eqref{eq:controlLawPre}.

In \sect \ref{sec:simulationStudy}, a simulation study of an IoT control problem with multiple agents is discussed.  
%Before, we present experimental results with a pair of agents as in \fig \ref{fig:agentSchematic}.
%, we next present experimental results for the remote estimation case as in \fig \ref{fig:agentSchematic}.

\subsection{Closed-loop analysis}
\label{sec:closedLoopAnalysis}
While the main object of study in this article are predictive and self triggering for state estimation (\cf \fig \ref{fig:remoteEstimation}), an important use of the algorithms is for feedback control as in \fig \ref{fig:iotControlSchematic} and \ref{fig:agentSchematic}.  
The general suitability of the algorithms for feedback control has already been demonstrated in \sect \ref{sec:pend_feedback}.
As for feedback control, analyzing the closed-loop dynamics (\eg for stability) is often of importance,  we briefly outline here how this can be approached.

The closed-loop state dynamics of agent $i$ are obtained from \eqref{eq:sys_x} and \eqref{eq:controlLawPre}, and can be rewritten as
\begin{align}
\xik{i}{k} &= A_{i} \xik{i}{k-1} + B_{i} F_i \xKFik{i}{k-1} + \!\!\!\! \sum_{j \in \N_N \setminus \{i\}}  \!\!\!  B_{i} F_j \xPik{j}{k-1} + \vik{v}{i}{k-1} \nonumber \\
&= A_{i} \xik{i}{k-1} + B_{i} F_i \xik{i}{k-1} + \!\!\!\! \sum_{j \in \N_N \setminus \{i\}}  \!\!\!  B_{i} F_j \xik{j}{k-1} \nonumber \\
&\phantom{=} - B_{i} F_i \vik{\eKF}{i}{k-1} - \!\!\!\! \sum_{j \in \N_N \setminus \{i\}}  \!\!\!  B_{i} F_j \vik{e}{j}{k-1} + \vik{v}{i}{k-1}
\end{align}
where $\vik{\eKF}{i}{k}$ is the KF estimation error \eqref{eq:def_KFerror} and $\vik{e}{j}{k}$ the remote estimation error \eqref{eq:def_remoteError}.
%$\vik{e}{i}{k} = \xik{i}{k} - \xPik{i}{k}$ is the error from the remote prediction \eqref{eq:remoteEst}, and $\vik{\hat{e}}{i}{k} = \xik{i}{k} - \xKFik{i}{k}$ was introduced for estimation error of the Kalman filter \eqref{eq:KF4}.  
The combined closed-loop dynamics of $N$ things with concatenated state $\tilde{x}_k^\transp = [(\xik{1}{k})^\transp, (\xik{2}{k})^\transp, \dots, (\xik{N}{k})^\transp]$ can then be written as
\begin{equation}
\tilde{x}_k = (\tilde{A}+\tilde{B}\tilde{F}) \tilde{x}_{k-1} - \tilde{D} \tilde{\hat{e}}_{k-1} - (\tilde{B} \tilde{F}- \tilde{D}) \tilde{e}_{k-1} + \tilde{v}_{k-1}
\label{eq:multiAgentClosedLoop}
\end{equation}
where
\begin{align*}
\tilde{A} &:= \diag(A_1, \dots, A_N), & 
\tilde{B}^\transp &:= \begin{bmatrix} B_1^\transp & \dots & B_N^\transp \end{bmatrix} \! , \\
\tilde{D} &:= \diag(B_1 F_1, \dots, B_N F_N), &
\tilde{F} &:= \begin{bmatrix} F_1 & \dots & F_N \end{bmatrix} \! ,
\end{align*}
%\begin{align*}
%A &= \begin{bmatrix} A_1 & & 0 \\
%& \ddots & \\
%0 & & A_N \end{bmatrix}, \,
%%A &= \diag(A_1, \dots, A_N), \quad 
%B = \begin{bmatrix} B_1 \\ \vdots \\ B_N \end{bmatrix}, \,
%F = \begin{bmatrix} F_1 & \dots & F_N \end{bmatrix}
%\end{align*}
$\diag$ denotes block-diagonal matrix, and 
$\tilde{\hat{e}}_{k}$, $\tilde{e}_{k}$, and $\tilde{v}_{k}$ are the combined vectors of all $\vik{\hat{e}}{i}{k}$, $\vik{e}{i}{k}$, and $\vik{v}{i}{k}$ ($i \in \N_N$), respectively.  
The `tilde' notation indicates variables that refer to the ensemble of all agents.

Equation \eqref{eq:multiAgentClosedLoop} describes the closed-loop dynamics of $N$ things of \fig \ref{fig:iotControlSchematic} that implement the control architecture in \fig \ref{fig:agentSchematic}; it can therefore be used to deduce closed-loop system properties.  
The evolution of the complete state $x_k$ is governed by the transition matrix $\tilde{A}+\tilde{B}\tilde{F}$ and driven by three input terms: the KF error $\tilde{\hat{e}}_{k-1}$, the remote error $\tilde{e}_{k-1}$, and process noise $\tilde{v}_{k-1}$.  
Under mild assumptions, the feedback matrix $\tilde{F}$ can be designed such that a stable transition matrix $\tilde{A}+\tilde{B}\tilde{F}$ results (\ie all eigenvalues with magnitude less than 1), which implies that $\tilde{x}_k = (\tilde{A}+\tilde{B}\tilde{F}) \tilde{x}_{k-1}$ is exponentially stable.  Stability analysis then amounts to showing that the input terms are well behaved and bounded in a stochastic sense (\eg bounded moments).\footnote{For example, if, in $\tilde{x}_k = (\tilde{A}+\tilde{B}\tilde{F}) \tilde{x}_{k-1} + \tilde{z}_{k-1}$, the input $\tilde{z}_k$ is uncorrelated and Gaussian with bounded variance, then stability of $\tilde{A}+\tilde{B}\tilde{F}$ implies bounded state variance (see, \eg \cite[\sect 4.3]{AnMo05}).}
While $\tilde{v}_{k-1}$ is Gaussian by assumption (\cf \sect \ref{sec:processDynamics}), $\tilde{\hat{e}}_{k-1}$ being Gaussian follows from standard KF analysis \cite{AnMo05} (\cf \sect \ref{sec:stateEstimation}).  Lemmas  \ref{lem:PDF_eII} and \ref{lem:PDF_eI} can be instrumental to analyze the distribution of $\tilde{e}_{k-1}$.  However, the distribution of $\tilde{e}_{k-1}$ depends on the chosen trigger, and its properties (\eg bounded second moment) would have to be formally shown, which is beyond the goals of this article.
%; in particular, they express mean and variance relative the .  
%The actual distribution depends on the chosen trigger.   

%% file: content/8_Simulations.tex
\section{Simulation Study: Vehicle Platooning}
\label{sec:simulationStudy}

To illustrate the scalability of the proposed triggers for IoT control, we
present a simulation study of vehicle platooning. Connected vehicles are seen as
a building block of the Internet of Vehicles~\cite{LuChengIoT}. 
%The framework presented in this paper shall be used
%in a simulation model for vehicle platooning.
%
Platooning of autonomous vehicles has been extensively studied in
literature, \eg for heavy-duty freight
transport~\cite{JohProc,AlamJoh2015}. It has been shown that
platooning leads to remarkable improvements in terms of fuel consumption.  

\subsection{Model}
We consider a chain of $N$ vehicles (see \fig \ref{fig:PlatooningSchematic}), which are modeled as unit point masses (\cf
\cite{LevAth,MuTr18}).
The state of each vehicle is its absolute position $s_i$ and velocity $v_i$, and its acceleration $u_i$ is the control input.
The control objectives are to maintain a desired distance between the vehicles and track a desired velocity for the platoon.
For this study, we assume that every vehicle measures its absolute position. 
%Therefore the state of the continuous-time model
%can be written as
%\begin{align}
%\label{eqn:state_veh_loc}
%x_i(t) &= \begin{pmatrix}
%s_i(t)\\
%v_i(t)
%\end{pmatrix} 
%\end{align} 
%with $i = 1,2,\ldots,M-1$
%and $t\in [0,\infty)$. The variables $s_i$ and $v_i$ denote the position and the
%velocity of vehicle $i$ respectively.
%The input $u_i$ to each vehicle is the acceleration.

\begin{figure}
\centering
\includegraphics{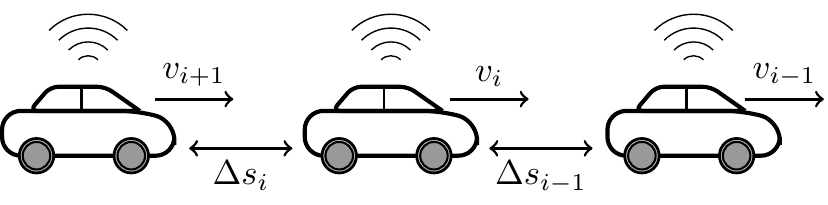}
\caption{Schematic of vehicle platooning.}
\label{fig:PlatooningSchematic}
\end{figure}

The architecture of the vehicle platoon is as in \fig \ref{fig:iotControlSchematic}.  
%Every vehicle is supposed to have a constant distance to its preceding vehicle,
%but only measures its own absolute position. 
To control the inter-vehicle distances, 
%while each vehicle only knows its own position, 
communication between the vehicles is required.  We thus implement the IoT control 
architecture given by \fig \ref{fig:agentSchematic} with PT and ST to save communication.
%the state information of the preceding vehicle has to be known. Thus, a
%communication network is necessary. 
We assume \SI{100}{\milli\second} as the sample time for the inter-vehicle
communication. 
Here, we consider the case where each vehicle transmits its local state information
to all other vehicles.  Alternative architectures, where communication is only
possible with a subset of vehicles, are also conceivable in the considered
scenario (see \cite{AlamJoh2015}), and the PT and ST can be used for only the
required communication links appropriately.
%Every vehicle transmits
%its state information to all others. Similarly to the examples shown before,
%communication shall be saved by predicting the states of the other vehicles
%based on the model equations and only communicating if necessary.

For our chosen setup, where each vehicle is only able to measure its own absolute position, it is obvious that communication between vehicles is necessary to control the inter-vehicle distance.
However, even if local sensor measurements are available, \eg if every vehicle can measure the distance to the preceding vehicle via a radar sensor, communication is required to guarantee \emph{string stability}.
String stability indicates whether oscillations are amplified upstream the traffic flow.
In~\cite{naus2010string}, it has been proven that if only local sensor measurements are used, string stability can only be guaranteed for velocity dependent spacing policies, \ie the faster the cars drive the larger distances are required, and thus, the less fuel can be saved.
Therefore, even in the presence of local measurements, communication between vehicles is crucial for fuel saving.
In such a case, where additional local sensor measurements are available, predictive and self triggering can similarly be used, as also stated in Remark~\ref{rem:localSensors}.

To address the control objectives, we design an LQR for the linear state-space model that includes
the vehicle velocities and their relative distances, \ie $x_{i}(t) = [v_i(t), \, s_i(t) - s_{i-1}(t)]^\transp$.
%\begin{align}
%x_{i}(t) = [v_i(t), \, s_i(t) - s_{i-1}(t)]^\transp
%\end{align}
The complete state $\tilde{x}$ is given by $x_{1}, x_{2}, \dots, x_N$ except for no relative position for the last vehicle $i=N$ (\cf \fig \ref{fig:PlatooningSchematic}).
For this system, an LQR is designed with $Q=I$ and $R= 1000 I$. The
even-numbered diagonal entries of the $Q$ matrix specify the
inter-vehicle distance tracking, while the odd ones weight the desired velocity.
To achieve tracking of desired velocity and inter-vehicle distance, the desired
state $\tilde{x}_\mathrm{des}$ is introduced, and the LQR law
$\tilde{u}_k = \tilde{F} (\tilde{x}_k - \tilde{x}_{\mathrm{des},k})$ implemented.

%The state of the global system used to derive the controller gains is
%different from the one in equation~\eqref{eqn:state_veh_loc}. The absolute
%positions are not controlled, but the relative distance between the vehicles.
%Therefore the state can be written as (\cf Fig.~\ref{fig:PlatooningSchematic})
%\begin{align}
%x_{i,\mathrm{ctrl}}(t) = \begin{pmatrix}figure
%v_i(t)\\s_i(t) - s_{i-1}(t)
%\end{pmatrix}.
%\end{align}
%For this system an LQR is designed, using for both the $Q$ and the $R$ matrix
%identity matrices, while the $R$ matrix is multiplied by \num{1000}. The
%odd-numbered diagonal entries of the $Q$ matrix define, how fast the velocity
%tends to zero, the even ones do the same for the inter-vehicle distance. As
%neither the velocity nor the inter-vehicle distance should go to zero, a desired
%state $x_\mathrm{des}$ is introduced. The input is then calculated via
%$u = K\left(x_\mathrm{ctrl} - x_\mathrm{des}\right)$.

We emphasize that the feedback gain matrix $\tilde{F}$ is dense; that is, information about all states 
in the platoon are used to compute the
optimal control input. Such controller can only be implemented in a distributed way, if complete state
information is available on each agent via the architecture presented in \sect \ref{sec:commNetwork}
 with all-to-all communication.

In the simulations\footnote{The Python source code for the simulations is
available under \url{https://github.com/baumanndominik/predictive_and_self_triggering}.} below, position measurements are
corrupted by independent noise, uniformly distributed in $[\SI{-0,1}{\meter},\SI{0,1}{\meter}]$.
Likewise, the inputs are corrupted by uniform noise in
$[\SI{-0,1}{\meter\per\second\squared},\SI{0,1}{\meter\per\second\squared}]$.
Additionally, we assume \SI{10}{\percent} Bernoulli packet drops.

\subsection{Platooning on changing surfaces}
We investigate the performance versus communication trade-off achieved with PT and ST for 
platooning of \num{10} vehicles. Here, we are interested in the closed-loop performance that 
is achieved with the proposed architecture; hence, instead of the estimation error, we use the sum of the absolute value of the error between $\tilde{x}$ and $\tilde{x}_\mathrm{des}$, 
normalized by the state dimension and number of time steps,
%divided by the state dimension and averaged over the time steps, 
as performance metric $\tilde{J}$.\footnote{LQR cost as one alternative performance metric leads to similar insights, but may have higher variance.}
The platoon drives for \SI{25}{\second}, while keeping desired inter-vehicle
distances of \SI{10}{\meter} and velocity of \SI{22,2}{\meter\per\second}.  After \SI{200}{\meter}, the dynamics change due to different road conditions
(\eg continue driving on a wet road after leaving a tunnel), which is 
modeled by altering the vehicle dynamics accordingly (vehicles moving
\SI{50}{\percent} faster, and the effect of braking/accelerating is reduced by
\SI{50}{\percent}).  Fig. \ref{fig:commVSerr_20veh_vel80} shows the results
from \num{100} Monte Carlo simulations.

%For \num{10} vehicles platooning with an inter-vehicle distance of
%\SI{10}{\meter} and a desired velocity of \SI{22,2}{\meter\per\second}, the
%tradeoff between the optimisation cost of PT and
%ST and normalized communication is shown in
%Fig.~\ref{fig:commVSerr_20veh_vel80}. For this example it was assumed, that
%after \SI{200}{\meter} the dynamics change. This might be the case, if the
%vehicles leave a tunnel and continue driving on a wet road. The friction of the
%wet road is reduced compared to the road in the tunnel, so the vehicles move
%faster and the effect of braking and accelerating is smaller. Therefore it is
%assumed, that the $A$ matrix of system~\ref{eqn:state_veh_loc} is multiplied by
%\num{1,5} whereas the $B$ matrix is multiplied by \num{0,5}.

Both triggers achieve significant communication savings at only a 
mild decrease of control performance.
Similar to studies in previous sections, the PT performs better than the
ST for low communication rates, because it can react to changing conditions.
For high communication rates, PT and ST are identical. If the prediction horizon is extended,
the performance of the PT gets closer to that of the ST, as can be obtained from the blue curve in Fig.~\ref{fig:commVSerr_20veh_vel80}.

%In Fig.~\ref{fig:commVSerr_20veh_vel80} every marker represents the mean of
%\num{1000} Monte Carlo simulations.

\begin{figure}
\centering
\includegraphics{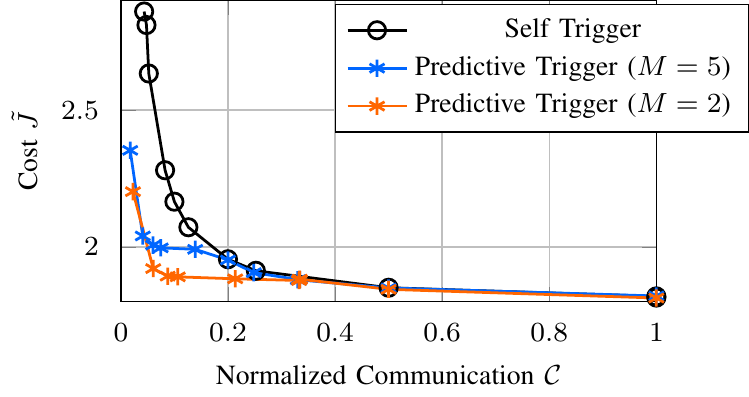}
\caption{Trade-off between normalized communication and control cost
for a \num{10} vehicles platoon.
%ing with a desired velocity of \SI{22,2}{\meter\per\second}.
Every marker represents the mean of \num{100} Monte Carlo simulations. The
variance is negligible and hence omitted. The plot shows the ST (black) as well
as two curves for the PT, one with a prediction horizon of \num{2} (orange) and one
with a prediction horizon of \num{5} (blue).}
\label{fig:commVSerr_20veh_vel80}
\end{figure}

\subsection{Braking}
If vehicles drive in close proximity, 
%If communication is limited in an example like this,
the ability to react to sudden changes, such as a braking maneuver of the preceding car, is critical. This is investigated here for three vehicles (simulation with more vehicles leads to the same insight).

Figure \ref{fig:80to0pred10} shows simulation results, where all cars start with
a velocity of \SI{22,2}{\meter\per\second}, but after \SI{10}{\second}, the first car brakes. 
%This might happen because the first vehicle detects the start of a
%traffic jam.
The results in Fig.~\ref{fig:80to0pred10} (left) show that even with very little
communication, the PT is able to deal with this situation.  The PT detects the need for more communication and is able to control inter-vehicle distances within safety bounds.
%The PT currently monitors the actual state and compares it to
%the predicted state. This error increases as the first vehicle starts to brake and
%such the communication rate increases as well. 
As previously pointed out, the ST (\fig \ref{fig:80to0pred10} right) cannot
react online, which causes a crash in this example ($\Delta s_1 = 0$).
%does not
%take this information into account, thus it does not react on the increasing
%error between the actual state and the predicted state and the first and second
%vehicle crash.

\begin{figure}
\centering
%\begin{minipage}{0.175\textwidth}
%\input{Tikz/80to0pred10.tex}
\includegraphics{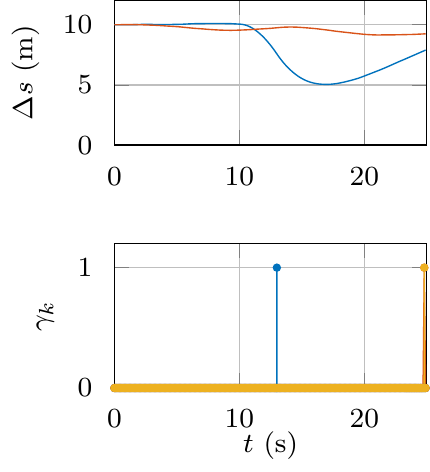}
%\end{minipage}
\hspace{0.1cm}
%\begin{minipage}{0.175\textwidth}
%\input{Tikz/80to0self04.tex}
\includegraphics{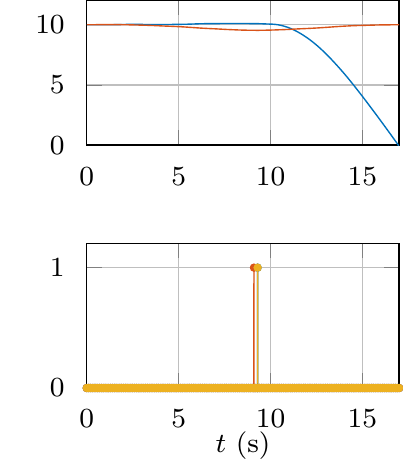}
%\end{minipage}
\caption{Three vehicles platooning with a constant velocity of
\SI{22,2}{\meter\per\second}. After \SI{10}{\second} the first car starts
braking. The top plot shows the distances
$\Delta s_1$ (blue) and $\Delta s_2$ (red); the bottom plot shows the
communication instants (vehicle 1 in blue, vehicle 2 in red, and vehicle 3 in
yellow).
The left plots show the behavior for the PT (with communication cost $C_k = C =
\num{10}$), the right plots for the ST (with communication cost $C_k = C = \num{0.7}$).
}
\label{fig:80to0pred10}
\end{figure}

% \begin{figure}
% \centering
% \input{Tikz/80to0self04.tex}
% \caption{\num{3} vehicles starting with a constant velocity of
% \SI{22,2}{\meter\per\second} and after \SI{10}{\second} braking down to
% \SI{0}{\meter\per\second}. The top plot shows the distances $\Delta s_1$ (blue)
% and $\Delta s_2$ (read), the bottom plot shows when the vehicles are communicating their state
% to the others. As the ST is used here and all vehicles have the same
% parameters, they always communicate at the same time.}
% \label{fig:80to0self04}
% \end{figure}

%For the sake of clarity of the figures these simulations were carried out with
%only \num{3} vehicles, but simulations with \num{10} vehicles lead to the same
%result.

%For an arbitrary large number of vehicles it is not convenient to assume, that
%every vehicle communicates its state information to all other vehicles. For the
%setting chosen here this is not necessary, it is sufficient if every vehicle
%sends its state information to a subset of surrounding
%vehicles~\cite{AlamJoh2015}. The only disadvantage of this is, that if for instance the first vehicle changes the
%desired velocity, this information needs to propagate through the platoon and is not available for
%all other vehicles at the next communication instant.

%% file: content/9_Conclusion.tex
\section{Conclusion}
\label{sec:conclusion}
In IoT control, feedback loops are closed between multiple things over a
general-purpose network.  Since the network is shared by many entities, communication is a limited resource that must be taken into account for optimal system-level operation when making control decisions.
% when making control decision.  
%For optimal system-level operation,   
This work sets a foundation for such 
%integrating communication and control and thus 
resource-aware IoT control. 
Distributed event-based state estimation (DEBSE) provides a powerful architecture for sharing information between multiple things and their cooperative control.  The developed self trigger and predictive trigger allow one to anticipate future communication needs, which is fundamental for efficiently (re-)allocating network resources.  

In order to leverage the potential of this work and realize actual resource savings on concrete IoT systems, the integration of ST and PT herein with a suitable communication system is essential.  While DEBSE has successfully been implemented on wired CAN bus networks in prior works \cite{TrDAn11,Tr17}, we target the integration with modern wireless network protocols such as the \emph{Low-power Wireless Bus} (LWB) \cite{FerZimLWB} in ongoing work.  LWB essentially abstracts a multi-hop wireless network as a common bus enabling fast \cite{zimmerling17} and reliable \cite{ferrari11} many-to-all communication.  Hence, it is ideally suited for scenarios such as in Figures \ref{fig:iotControlSchematic} and \ref{fig:agentSchematic}, where multiple things require information about each other for coordination.  In particular, all-to-all communication allows for the effective realization of the predictors \eqref{eq:remoteEst} on any agent that needs the corresponding state information. LWB typically runs a network manager on one of its nodes, 
%which operates as global network manager.  This manager 
which can use the communication requirements signaled by ST and PT to schedule next communication rounds.  The concrete development and integration of such schemes is subject of ongoing research.
%
% Previous (much shorter):
%The actual integration of the distributed and event-based control system with network protocols such as \cite{FerZimLWB} is ongoing research.  
While the focus of this article is on saving communication bandwidth, the proposed triggers can also be instrumental for saving other resources in IoT (\eg computation or energy).

The predictive and self triggers are suitable for different application scenarios.  The simulation and experimental studies herein clearly highlight the advantage of the predictive trigger: by continuously monitoring the triggering condition, it can react to unforeseeable events such as disturbances.  The self trigger, on the other hand, is an offline trigger, which allows for setting devices to sleep.
% for saving energy.  
In contrast to commonly used event triggers, both proposed triggers can \emph{predict} resource needs rather than making instantaneous decisions.
Predictive triggering is a novel concept in-between the previously proposed concepts of self triggering and event triggering.
%, as is shown in the analysis and simulation results herein.

Concrete instances of the predictive and self trigger were derived herein for estimation of linear Gaussian systems.  
While the general idea of predicting triggers also extends to nonlinear estimation, properly formalizing this and deriving triggering laws for nonlinear problems is an interesting task for future work.  
%The ideas, however, directly extend also to nonlinear estimation problems.  Deriving triggering laws for such problems is an interesting task for future work.  
Likewise, considering alternative optimization problems for different error choices in \eqref{eq:Ek_squares}, as well as dynamic programming formulations in place of the one-step optimization in \eqref{eq:optProblET_ideal}, may lead to interesting insights and alternative triggers.
While the predictive and self triggers herein were shown to stabilize the inverted pendulum in the reported experiments, formally analyzing stability of the closed-loop system (\eg along the lines outlined in \sect \ref{sec:closedLoopAnalysis})
%(especially with respect to the length of the prediction horizon) 
is another relevant open research question.

%% file: content/Proof2.tex
%%%%%%%%%%%%%%%%%%%%%%%%%%%%%%%%%%%%%%%%%%%%%%%%%%%%%%%%%%%%%%%%%%%%%%%%%%%%%%%%%%
%We first prove \Lem \ref{lem:PDF_eII}, which will be used in the proof of \Lem \ref{lem:PDF_eI}.

Because $\xP_k = \xKF_k$ for $\gamma_k=1$ from \eqref{eq:remoteEst},
%(\ie $\xII_k$ is reset to the estimate by the KF with full communication \eqref{eq:KF1}--\eqref{eq:KF5}), 
the remote error $e_k$ is identical to the KF error $\eKF_k= x_k - \xKF_k$.
% of the standard KF \eqref{eq:KF1}--\eqref{eq:KF5}.  
From KF theory \cite[p.~41]{AnMo05}, it is known that the conditional and unconditional error distributions are identical, namely
\begin{equation}
f(\eKF_k) = f(\eKF_k | \Yall_k, \Uall_k) = \Nc(\eKF_k; 0, \PKF_k) .
\label{eq:lem2_eKF}
\end{equation}
That is, the error distribution is independent of any measurement data.  
% Short version:
%Therefore, also $f(\eKF_{k+M} | \Yall_k) = f(\eKF_{k+M})$, from which the claim follows.
%
% CUT: (re-include if possible)
%
Therefore, we also have $f(e_{k+M} | \Yall_k, \Uall_k) = f(\eKF_{k+M} | \Yall_k, \Uall_k) = f(\eKF_{k+M})$ (see \cite[Proof of Lem.~2]{Tr16} for a formal argument), from which the claim follows with \eqref{eq:lem2_eKF}.

%% file: content/Proof1.tex
%%%%%%%%%%%%%%%%%%%%%%%%%%%%%%%%%%%%%%%%%%%%%%%%%%%%%%%%%%%%%%%%%%%%%%%%%%%%%%%%%%
%
%%%%%%%%%%%%%%%%%%%%%%%%%%%%%%%%%%%%%%%%
We first establish, for any $M \geq 0$,
\begin{align}
\xKF_{k+M} &= \bar{A}^M \xKF_k 
+ \sum_{m=1}^{M} \bar{A}^{M-m}  B \xi_{k+m-1} \nonumber \\
&\phantom{=} + \sum_{m=1}^{M} \bar{A}^{M-m}  L_{k+m} z_{k+m} \label{eq:induction1} \\
\xKF_{k+M|k} &= \bar{A}^M \xKF_k + \sum_{m=1}^{M} \bar{A}^{M-m}  B \xi_{k+m-1} \nonumber \\
&\phantom{=}+ \sum_{m=1}^{M-1} G_{M-m-1} L_{k+m} z_{k+m} \label{eq:induction2}
\end{align}
with $z_k := y_k - H \xKF_{k|k-1}$ the KF innovation, $L_k$ the KF gain, and $G_m$ as in \eqref{eq:G_def}, through proof by induction.  For $M=0$, \eqref{eq:induction1} and \eqref{eq:induction2} hold trivially with $\xKF_k = \xKF_k$ and $\xKF_{k|k} = \xKF_k$, respectively.
% (sums understood to vanish if $m$ not increasing).  
Induction assumption (IA): assume \eqref{eq:induction1} and \eqref{eq:induction2} hold for $M$. Show they are then also true for $M + 1$.  We have from the KF iterations: 
\begin{align*}
\xKF_{k+M+1} &= A\xKF_{k+M} + B u_{k+M} + L_{k+M+1} z_{k+M+1} \\
&= \bar{A} \xKF_{k+M} + B \xi_{k+M} +  L_{k+M+1} z_{k+M+1} \;\; \text{(by \eqref{eq:controlLaw})} \\
&= \bar{A}^{M+1} \xKF_k + \sum_{m=1}^{M+1} \bar{A}^{M+1-m}  B \xi_{k+m-1} \nonumber \\
&\phantom{=} +\sum_{m=1}^{M+1} \bar{A}^{M+1-m}  L_{k+m} z_{k+m} \quad \text{(from IA \eqref{eq:induction1})}
\end{align*}
and
\begin{align*}
&\xKF_{k+M+1|k} = A \xKF_{k+M|k} + B u_{k+M} \\
&\phantom{=}= A \xKF_{k+M|k} + BF \xKF_{k+M} + B \xi_{k+M} \\
&\phantom{=}= (A+BF) \Big(\bar{A}^M \xKF_k 
+ \sum_{m=1}^{M} \bar{A}^{M-m}  B \xi_{k+m-1} \Big) + B \xi_{k+M} \\
&\phantom{=}\phantom{=} 
    + A \Big( \sum_{m=1}^{M-1} G_{M-m-1} L_{k+m} z_{k+m} \Big) \\
&\phantom{=}\phantom{=}    
    + BF \Big( \sum_{m=1}^{M} \bar{A}^{M-m}  L_{k+m} z_{k+m} \Big) \quad \text{(from IA \eqref{eq:induction1}, \eqref{eq:induction2})} \\
&\phantom{=}= \bar{A}^{M+1} \xKF_k + \sum_{m=1}^{M+1} \bar{A}^{M+1-m}  B \xi_{k+m-1}  \\
&\phantom{=}\phantom{=}  
    + \sum_{m=1}^{M} G_{M-m} L_{k+m} z_{k+m} \quad \text{(by def.\ of $G_m$)}.
\end{align*}
Hence, \eqref{eq:induction1} and \eqref{eq:induction2} are true for $M+1$, which completes the induction.

Next, we analyze the error $e_{k+M}$ for the case $\gamma_{k+M}=0$ (no communication).  To ease the presentation, we introduce the auxiliary variable $\eIaux_k := e_{k}\notrigvar{k}$.

{\it Case (i):}
%Now, we turn to the proof of \emph{Case (i)}.  
First, we note that $k > \lastel_{k-1}$ implies $\lastel_{k-1} = \last_{k}$ because $\lastel_{k-1}$, the last nonzero element of $\Gamall_{k+M-1}$, is in the past, 
%(less then $k$) 
and the identity thus follows from the definition of $\last_{k}$.  It follows further that all triggering decisions following $\gamma_\last = 1$ are 0 
%for all steps after $\last$ 
until $\gamma_{k+M-1}$ (otherwise $\gamma_\last$ would not be the last element in $\Gamall_{k+M-1}$).  Hence, we have the communication pattern $\gamma_\last = 1$ and $\gamma_{\last+1} = \gamma_{\last+2} = \dots = \gamma_{k+M-1}=0$.

Let $\tilde\Delta:= M+k-\ell$.  From 
\begin{equation*}
\eIaux_{k+M} 
%e_{k+M}\notrigvar{k+M}
= x_{k+M} - \bar{A}^{\tilde{\Delta}} \xKF_\ell - \sum_{m=1}^{\tilde{\Delta}} \bar{A}^{\tilde{\Delta}-m}  B \xi_{\ell+m-1}
\end{equation*}
it follows that the conditional distribution \eqref{eq:lem1_eIpdf} is Gaussian.  It thus suffices to consider mean and variance in the following.

For the conditional mean, we have
\begin{align}
&\E[\eIaux_{k+M} | \Yall_k, \Uall_k] \nonumber\\
%&\E[(e_{k+M}\notrigvar{k+M}) | \Yall_k, \Uall_k] \nonumber\\
&= \E[ x_{k+M} | \Yall_k, \Uall_k] - \bar{A}^{\tilde{\Delta}} \xKF_\ell - \sum_{m=1}^{\tilde{\Delta}} \bar{A}^{\tilde{\Delta}-m}  B \xi_{\ell+m-1} , \label{eq:lem1_proof_1}
\end{align}
and
\begin{align}
\E[ x_{k+M} | \Yall_k, \Uall_k]
&= \E\big[ \E[x_{k+M} | \Yall_k, \Uall_{k+M}] \big| \Yall_k, \Uall_k \big] \nonumber \\
&= \E[ \xKF_{k+M|k} | \Yall_k, \Uall_k] \nonumber \\
&= \bar{A}^M \xKF_k + \sum_{m=1}^{M} \bar{A}^{M-m}  B \xi_{k+m-1} \label{eq:lem1_proof_2}
\end{align}
where we used the tower property of conditional expectation, \eqref{eq:PDFstatePredM}, and \eqref{eq:induction2} with the fact that the KF innovation sequence $z_k$ is zero-mean and uncorrelated. Using \eqref{eq:lem1_proof_2} with \eqref{eq:lem1_proof_1}, we obtain
\begin{align}
&\E[\eIaux_{k+M} | \Yall_k, \Uall_k]
= \bar{A}^M (\xKF_k - \bar{A}^{k-\ell}  \xKF_\ell) +\sum_{m=1}^{M} \bar{A}^{M-m} B \xi_{k+m-1} \nonumber \\
&\phantom{=}    - \sum_{m=1}^{k-\ell} \bar{A}^{\tilde{\Delta}-m}  B \xi_{\ell+m-1} - \sum_{m=k-\ell+1}^{M+k-\ell} \bar{A}^{M+k-\ell-m}  B \xi_{\ell+m-1} \label{eq:lem1_proof_3} \\
&= \bar{A}^M \Big(\xKF_k - \bar{A}^{k-\ell}  \xKF_\ell - \sum_{m=1}^{k-\ell} \bar{A}^{k-\ell-m}  B \xi_{\ell+m-1} \Big)
\end{align}
which proves \eqref{eq:lem1_eImean}.  The first and third sum in \eqref{eq:lem1_proof_3} can be seen to be identical by substituting $m$ with $m + k - \ell$.

Employing the tower property for the conditional variance, we get
\begin{align*}
&\Var[\eIaux_{k+M} | \Yall_k, \Uall_k] \nonumber \\
&\qquad= \E\big[\Var[\eIaux_{k+M} | \Yall_k, \Uall_{k+M}] \big| \Yall_k, \Uall_k \big]  \nonumber \\
&\qquad\phantom{=}+ \Var\big[\E[\eIaux_{k+M} | \Yall_k, \Uall_{k+M}] \big| \Yall_k, \Uall_k \big] \nonumber \\
&\qquad= \E[\PKF_{k+M|k} | \Yall_k, \Uall_k] + \Var[\xKF_{k+M|k} | \Yall_k, \Uall_k] \nonumber \\
&\qquad= \PKF_{k+M|k} + \Var[\xKF_{k+M|k} | \Yall_k, \Uall_k].
\end{align*}
Furthermore, $\Var[\xKF_{k+M|k} | \Yall_k, \Uall_k] = \Xi_{k,M}$ follows from \eqref{eq:induction2}, $z_k$ being uncorrelated, and 
\begin{align*}
&\Var[z_{k+m} | \Yall_k, \Uall_k] \nonumber \\
&= \Var[HA \eKF_{k+m-1} + H v_{k+m-1} + w_{k+m} | \Yall_k, \Uall_k] \nonumber \\
&= \tilde{P}_{k+m}
\end{align*}
as defined in \eqref{eq:Ptilde_def}.  This completes the proof for \emph{Case (i)}.

%%%%%%%%%%%%%%%%%%%%%%%%%
{\it Case (ii):}
We use $\lastel = \lastel_{k-1}$ to simplify notation.
By definition of $\lastel$, we have $\lastel \leq M+k-1$, and hence $k \leq \lastel \leq M+k-1$.  That is, a triggering will happen now or 
%in future $k \leq \lastel$, and 
before the end of the horizon $M+k$.  At the triggering instant $\lastel$, we have from \eqref{eq:remoteEst}, $e_{\lastel} = x_{\lastel} - \xKF_{\lastel}$.
% CUT:
%\begin{equation}
%e_{\lastel} = x_{\lastel} - \xKF_{\lastel}.
%\end{equation}
%That is, at the time of communication $\lastel$, the remote estimate is reset to the KF estimate.  
Hence, the distribution of the error at time $\lastel$ is known irrespective of past and future data.  Following the same arguments as in the proof of \Lem \ref{lem:PDF_eII}, we have $f(e_{\lastel} | \Yall_k, \Uall_k) = f(e_{\lastel} | \Yall_\lastel, \Uall_\lastel) = \Nc(e_{\lastel}; \, 0, \PKF_{\lastel})$.

From the definition of $\lastel$, we know that there is no further communication happening until $M+k-1$.  Thus, we can iterate \eqref{eq:remoteEst} with $\gamma=0$.
Using the same reasoning as in \emph{Case (i)}, we have
\begin{equation*}
\eIaux_{k+M} = \eIaux_{\kappa + \Delta} = 
x_{\kappa + \Delta} - \bar{A}^\Delta \xKF_\kappa - \sum_{m=1}^{\Delta} \bar{A}^{\Delta-m}  B \xi_{\kappa+m-1}
\end{equation*}
and thus
\begin{align*}
&\E[\eIaux_{\kappa + \Delta} | \Yall_\kappa, \Uall_\kappa] \nonumber\\
&= \E[ x_{\kappa + \Delta} | \Yall_\kappa, \Uall_\kappa] -  \bar{A}^\Delta \xKF_\kappa - \sum_{m=1}^{\Delta} \bar{A}^{\Delta-m}  B \xi_{\kappa+m-1} \nonumber \\
&=\E[ \xKF_{\kappa + \Delta|\kappa} | \Yall_\kappa, \Uall_\kappa] -  \bar{A}^\Delta \xKF_\kappa - \sum_{m=1}^{\Delta} \bar{A}^{\Delta-m}  B \xi_{\kappa+m-1} 
%\nonumber \\
%&
= 0
\end{align*}
where the last equality follows from \eqref{eq:induction2} and $z_k$ being zero-mean.
Similarly, for the variance, we obtain
\begin{align*}
\Var[\eIaux_{\kappa + \Delta} | \Yall_\kappa, \Uall_\kappa] &=
\E[\PKF_{\kappa + \Delta|\kappa} | \Yall_\kappa, \Uall_\kappa] + \Var[\xKF_{\kappa + \Delta|\kappa} | \Yall_\kappa, \Uall_\kappa] \nonumber \\
&= \PKF_{\kappa+\Delta|\kappa} + \Var[\xKF_{\kappa+\Delta|\kappa} | \Yall_\kappa, \Uall_\kappa] \nonumber \\
&= \PKF_{\kappa+\Delta|\kappa} + \Xi_{\kappa, \Delta}.
\end{align*}